\documentclass[submission,copyright,creativecommons]{eptcs}
\providecommand{\event}{TERMGRAPH 2022}

\usepackage[titletoc]{appendix}
\usepackage{stmaryrd}
\usepackage{amssymb}
\usepackage{amsmath}
\usepackage{graphicx}
\usepackage{enumerate}
\usepackage{enumitem}
\usepackage{mathtools}
\usepackage{multicol}
\usepackage{subcaption}
\usepackage{proof}
\usepackage{amsthm}
\usepackage[normalem]{ulem}
\usepackage{array}

\theoremstyle{plain}
\newtheorem{lemma}{Lemma}
\newtheorem{theorem}{Theorem}
\newtheorem{proposition}{Proposition}

\theoremstyle{definition}

\newtheorem{definition}{Definition}
\newtheorem{construction}{Construction}

\theoremstyle{remark}
\newtheorem{example}{Example}
\newtheorem{remark}{Remark}

\usepackage{tikz}
\usepackage{tikz-cd}
\usetikzlibrary{patterns,arrows, topaths, calc, positioning}
\tikzset{>=stealth}
\tikzstyle{node} = [circle, minimum size = 1.4mm, inner sep = 0mm, color=black, fill]
\tikzstyle{hyperedge} = [rectangle, minimum width = 5mm, minimum height = 5mm, draw, inner sep = 0mm]
\tikzstyle{HG} = [align = center]
\tikzstyle{circledge} = [circle, minimum size = 7mm, inner sep = 0mm, color=black, draw]

\pgfdeclarepatternformonly{NELines}{\pgfqpoint{-1pt}{-1pt}}{\pgfqpoint{4pt}{4pt}}{\pgfqpoint{3pt}{3pt}}%
{
	\pgfsetlinewidth{0.4pt}
	\pgfpathmoveto{\pgfqpoint{0pt}{0pt}}
	\pgfpathlineto{\pgfqpoint{1pt}{1pt}}
	\pgfpathmoveto{\pgfqpoint{2pt}{2pt}}
	\pgfpathlineto{\pgfqpoint{3.1pt}{3.1pt}}
	\pgfusepath{stroke}
}

\newcommand{\eqdef}{\mathrel{\mathop:}=}

\newcommand{\LC}{\mathrm{L}}
\newcommand{\I}{\mathrm{\mathbf{1}}}
\newcommand{\Dis}{\mathrm{D}}
\newcommand{\HL}{\mathrm{HL}}
\newcommand{\HMEL}{\mathrm{HMEL}_0}
\newcommand{\MILLFO}{\mathrm{MILL1}}
\newcommand{\ILLFO}{\mathrm{ILL1}}

\newcommand{\LIKS}{\mathrm{L}_{\I\omega}^\ast}

\newcommand{\DPO}{\mathrm{DPO}}

\title{From Double Pushout Grammars to Hypergraph Lambek Grammars With and Without Exponential Modality}
\author{Tikhon Pshenitsyn
	\institute{Department of Mathematical Logic and Theory of Algorithms, Faculty of Mathematics and Mechanics
	\\
	Lomonosov Moscow State University, GSP-1, Leninskie Gory, Moscow, 119991, Russian Federation}
	\thanks{The study was funded by RFBR, project N20-01-00670; the Interdisciplinary Scientific and Educational School of Moscow University ``Brain, Cognitive Systems, Artificial Intelligence''; the Theoretical Physics and Mathematics Advancement Foundation ``BASIS''.
	}
	\email{ptihon at yandex.ru}
}
\def\authorrunning{T. Pshenitsyn}
\def\titlerunning{From DPO Grammars to HL-Grammars With and Without Exponential Modality}

\begin{document}
	\maketitle
\begin{abstract}
	We study how to relate well-known hypergraph grammars based on the double pushout (DPO) approach and grammars over the hypergraph Lambek calculus $\HL$ (called $\HL$-grammars). It turns out that DPO rules can be naturally encoded by types of $\HL$ using methods similar to those used by Kanazawa for multiplicative-exponential linear logic. In order to generalize his reasonings we extend the hypergraph Lambek calculus by adding the exponential modality, which results in a new calculus $\HMEL$; then we prove that any DPO grammar can be converted into an equivalent $\HMEL$-grammar. We also define the conjunctive Kleene star, which behaves similarly to this exponential modality, and establish a similar result. If we add neither the exponential modality nor the conjunctive Kleene star to $\HL$, then we can still use the same encoding and show that any DPO grammar with a linear restriction on the length of derivations can be converted into an equivalent $\HL$-grammar. 
\end{abstract}

\section{Introduction}
In this paper, we aim to relate two kinds of graph grammars: double-pushout (DPO) hypergraph grammars and hypergraph Lambek grammars ($\HL$-grammars). 

DPO hypergraph grammars are one of the most well-known kinds of graph grammars, which were introduced in 1973 \cite{EhrigPS73}. They are designed to generalize unrestricted Chomsky formal grammars from strings to graphs. Recall that a production in a formal grammar of the form $\alpha \Rightarrow \beta$ allows one to replace a substring $\alpha$ in any string $\gamma$ by a string $\beta$. A production of a DPO hypergraph grammar, in turn, can be presented in the form $L \Rightarrow R$ where $L$ and $R$ are two hypergraphs. The procedure of replacing a hypergraph by another hypergraph, however, needs further clarification; this is done by using the double pushout approach, which is widely used in the field of graph grammars.

The hypergraph Lambek calculus $\HL$ and hypergraph Lambek grammars are novel approaches described in \cite{Pshenitsyn21_2, Pshenitsyn22}. They are based on logical grounds: $\HL$ generalizes the Lambek calculus introduced in \cite{Lambek58}. The Lambek calculus $\mathrm{L}$ is a substructural logic of intuitionistic logic, and it is originally designed to model the syntax of natural languages. The hypergraph Lambek calculus $\HL$ inherits the main principles of $\mathrm{L}$, its structural and model-theoretic properties. Besides, $\HL$ forms the basis for hypergraph Lambek grammars ($\HL$-grammars). An $\HL$-grammar is defined by an assignment of a finite number of types (i.e. formulas) of $\HL$ to symbols of an alphabet. With some simplifications, the mechanism of these grammars can be described as follows: in order to check that a hypergraph $H$ is generated by an $\HL$-grammar, one must replace each symbol in $H$ by one of the types corresponding to it (which results in a hypergraph labeled by types) and to check that the resulting structure is derivable from axioms by rules of $\mathrm{HL}$. 

Our objective is to figure out what class of hypergraph languages $\HL$-grammars generate. In particular, it is clearly important to compare them with widely studied DPO grammars. The following question arises: can we convert each DPO grammar into an $\HL$-grammar generating the same language? A simple complexity argument shows that the answer is negative: DPO grammars are Turing complete while the membership problem for $\HL$-grammars is decidable and even NP-complete. Nevertheless, it turns out that there is a simple way to naturally encode any DPO rule as a type of $\HL$. This encoding is essentially the same as the encoding that Kanazawa defined for translating unrestricted Chomsky grammars into grammars over the multiplicative-exponential Lambek calculus \cite{Kanazawa99}. In the hypergraph case, we can also add the exponential modality $!$ to $\HL$ obtaining the hypergraph multiplicative-exponential Lambek calculus; in this paper, we define it in a somehow restricted way, which helps us to prove the cut elimination theorem for it, yet sufficient for our purposes. The new calculus is denoted as $\HMEL$. Then we can show that any DPO grammar can be converted into an $\HMEL$-grammar. In fact, apparently $\HMEL$-grammars are equivalent to DPO grammars; however, in this paper, we prove only one half of this statement, so the other one formally remains a claim.

However, our main goal is still to study the hypergraph Lambek calculus itself without the exponential modality. It turns out that we can establish a nice relationship using the same encoding as for $\HMEL$-grammars; to do this, we need to restrict DPO grammars. Namely, let us equip a DPO grammar $Gr$ with a number $c \in \mathbb{N}$ and let us consider the language of hypergraphs $H$ generated by $Gr$ such that the length of a derivation of $H$ in $Gr$ is less than or equal to $c\cdot |E_H|$ where $|E_H|$ is the number of hyperedges in $H$. That is, we impose the linear restriction on the length of a derivation of a hypergraph with respect to the number of its hyperedges. It turns out that each such language can be generated by an $\HL$-grammar. In fact, the converse statement (from $\HL$-grammars to linearly restricted DPO grammars) holds as well, which is, however, a matter of another article to be written.

In the first version of this paper presented at TERMGRAPH 2022 in Haifa, we introduced an extension of $\HL$ by the \emph{conjunctive Kleene star} instead of the exponential modality. The conjunctive Kleene star behaves similar to the exponential modality $!$ (or, more precisely, to the soft subexponential as in \cite{KanovichKNS20}). However, the conjunctive Kleene star has a different motivation and its axiomatization involves, in particular, the omega rule (see Section \ref{sec_cks}). Such an operation is quite peculiar and weird, and it received little attention in the literature (however, it is mentioned e.g. in \cite{Montagna04} when considering \emph{storage operators}). At the workshop discussion, I was advised to try to use the exponential modality to achieve the same goals, which is more common and well studied. Then, when we started studying the existing literature on the multiplicative-exponential Lambek calculus for strings, we found the article \cite{Kanazawa99} by Kanazawa, where a similar construction is used to the one developed in the current work. So, in this final version of the paper, we do not use the conjunctive Kleene star in the constructions of grammars but the exponential modality. However, we retain the definition of the former in Section \ref{sec_cks} in order not to change the content of the paper too much.

The paper is organized as follows. Section \ref{sec_preliminaries} contains definitions of all formalisms and notions of interest (hypergraphs, hyperedge replacement, DPO grammars, $\HL$, $\HL$-grammars) along with some intuition behind the hypergraph Lambek calculus. In Section \ref{sec_hmel}, we define the hypergraph multiplicative-exponential Lambek calculus $\HMEL$ and generalize the idea from \cite{Kanazawa99} to prove that DPO grammars can be converted into equivalent $\HMEL$-grammars. In Section \ref{sec_lin_dpo_to_hl}, we impose the linear restriction on DPO grammars and present the translation procedure from restricted grammars to $\HL$-grammars. Section \ref{sec_cks} contains the definition of the conjunctive Kleene star and some remarks concerning it (although this operation is not used anymore in transformations). In Section \ref{sec_conclusion}, we conclude.

\section{Preliminaries}\label{sec_preliminaries}

$\Sigma^\ast$ is the set of strings over an alphabet $\Sigma$ including the empty word $\Lambda$; if $R$ is a relation, then $R^\ast$ is its transitive reflexive closure. Each function $f:\Sigma \to \Delta$ can be extended to a homomorphism $f:\Sigma^\ast \to \Delta^\ast$. By $w(i)$ we denote the $i$-th symbol of $w \in \Sigma^\ast$, and by $|w|$ we denote the number of symbols in $w$.

Let $[n]$ denote the set $\{1,2,\dotsc,n\}$ (and $[0] \eqdef \emptyset$ accordingly).

Given a set of \emph{labels} $\Sigma$ along with a \emph{rank function} $rk:\Sigma \to \mathbb{N}$, a \emph{hypergraph} $G$ over $\Sigma$ is a tuple $G=\langle V_G, E_G, att_G, lab_G, ext_G \rangle$ where $V_G$ is a finite set of \emph{nodes}, $E_G$ is a finite set of \emph{hyperedges}, $att_G: E_G\to V_G^\ast$ assigns a string (understand it as an ordered multiset) of \emph{attachment nodes} to each hyperedge, $lab_G: E_G \to \Sigma$ labels each hyperedge by some element of $\Sigma$ in such a way that $rk(lab_G(e))=|att_G(e)|$ whenever $e\in E_G$, and $ext_G\in V_G^\ast$ is a string of \emph{external nodes}.  Hypergraphs are always considered up to isomorphism. The set of all hypergraphs with labels from $\Sigma$ is denoted by $\mathcal{H}(\Sigma)$. Note that we allow attachment nodes of a hyperedge as well as external nodes to coincide. The \emph{rank function} $rk_G$ (or $rk$, if $G$ is clear) is defined as follows: $rk_G(e)\eqdef|att_G(e)|$. Besides, $rk(G)\eqdef |ext_G|$.

In drawings of hypergraphs, black circles correspond to nodes, labeled rectangles correspond to hyperedges, $att$ is represented by numbered lines, and external nodes are represented by numbers in parentheses (round brackets). If a hyperedge has exactly two attachment nodes, then it is depicted by a labeled arrow that goes from the first attachment node to the second one.

A \emph{handle} $a^\bullet$ is a hypergraph $a^\bullet=\langle [n],[1],att,lab,1\dotsc n\rangle$ where $att(1)=1\dotsc n$ and $lab(1)=a$ ($a\in\Sigma$, $rk(a)=n$). A hypergraph $a^\circ$ is of the form $\langle [n], [1], att,lab,\Lambda\rangle$ where $att$, $lab$ are as in the definition of $a^\bullet$. A hypergraph $\Dis[k] = \langle [k],\emptyset,\emptyset,\emptyset,\Lambda\rangle$ is called \emph{discrete} ($k \in \mathbb{N}$).

Given a hypergraph $H$ and a function $f:E_H\to \Sigma$, \emph{a relabeling $f(H)$} is the hypergraph $f(H)=\\\langle V_H, E_H, att_H, f, ext_H\rangle$. It is required that $rk_H(e)=rk(f(e))$ for $e\in E_H$. 

The \emph{replacement of a hyperedge $e_0$ in $G$ by a hypergraph $H$} (such that $rk(e_0)=rk(H)$) is done as follows: (1) remove $e_0$ from $G$; (2) insert an isomorphic copy of $H$ ($H$ and $G$ must consist of disjoint sets of nodes and hyperedges); (3) for each $i=1,\dotsc,rk(e_0)$, fuse the $i$-th external node of $H$ with the $i$-th attachment node of $e_0$. The result is denoted as $G[e_0/H]$. It is well known that if several hyperedges of a hypergraph are replaced by other hypergraphs, then the result does not depend on the order of the replacements; moreover the result is not changed, if replacements are done simultaneously \cite{Drewes97}. The following notation is in use: if $e_1,\dots,e_k$ are distinct hyperedges of a hypergraph $H$ and they are simultaneously replaced by hypergraphs $H_1,\dots,H_k$ resp., then the result is denoted $H[e_1/H_1,\dots,e_k/H_k]$.

In a case where a hypergraph $G$ does not have external nodes ($ext_G=\Lambda$) let us call it \emph{zero rank}. If at least one of the hypergraphs $H_1$, $H_2$ is zero rank, then one can define their \emph{disjoint union} $H_1+H_2$ as the hypergraph $\langle V_{H_1}\sqcup V_{H_2}, E_{H_1}\sqcup E_{H_2}, att, lab, ext\rangle$ such that $att|_{H_i} = att_{H_i}$, $lab|_{H_i} = lab_{H_i}$ ($i=1,2$), and $ext = ext_{H_i}$ if $H_i$ is not zero rank and $ext=\Lambda$ otherwise; that is, we just put these hypergraphs together without fusing any nodes or hyperedges. The disjoint union of a zero-rank hypergraph $H$ with itself $k$ times is denoted by $k \cdot H$. 

\subsection{DPO Grammars}
Given two hypergraphs $G$ and $H$, a \emph{morphism} $f:G \to H$ is a pair of functions $f_V:V_G \to V_H$, $f_E: E_G \to E_H$ such that $f_V(att_G(e))=att_H(f_E(e))$, $lab_H(f_E(e))=lab_G(e)$ for all $e\in E_G$, and $f_V(ext_G) = ext_H$. 

Let $I,G_1,G_2$ be zero-rank hypergraphs with morphisms $\varphi_i:I \to G_i$, $i=1,2$. Let $\equiv_V$ be the smallest equivalence relation on the disjoint union $V_{G_1} \sqcup V_{G_2}$ that satisfies $\varphi_1(v) \equiv \varphi_2(v)$ for $v \in V_I$; a relation $\equiv_E$ is defined similarly on $E_{G_1} \sqcup E_{G_2}$. $\langle x \rangle$ denotes the equivalence class of $x$ with respect to $\equiv_V$ if $x$ is a node, and with respect to $\equiv_E$ if $x$ is a hyperedge. The \emph{gluing} of $G_1$ and $G_2$ over $I$ denoted as $G_1 +_{\varphi_1,\varphi_2} G_2$ is a hypergraph $G$ such that $V_G=(V_{G_1} \sqcup V_{G_2})/\equiv_V$, $E_G=(E_{G_1} \sqcup E_{G_2})/\equiv_E$; given $\langle e \rangle \in E_G$ with $rk(e) = k$, if $e \in E_{G_i}$ for some $i=1,2$, then $att_G(\langle e \rangle) = \langle att_{G_i}(e)(1)\rangle \dotsc \langle att_{G_i}(e)(k)\rangle$ and $lab_G(\langle e \rangle) = lab_{G_i}(e)$. This is a well-defined notion taken from \cite{Konig18} (where it is defined for graphs rather than for hypergraphs). There, the authors state that the gluing of two graphs is a pushout in the category of graphs. In this paper, we do not work within the categorical approach, so we stick to the set-theoretic definition.

Note that, if $I$ is discrete, then the gluing procedure can be represented as replacement:
\begin{proposition}\label{prop_dpo_to_replacement}
	Let $I=\Dis[k]$ and let $G_i$, $\varphi_i$ be as above. Let $G_1^\prime = \langle V_{G_1},E_{G_1}\sqcup \{e_0\},att_{G_1^\prime},lab_{G_1^\prime}\rangle$ where $att_{G_1^\prime}(e) = att_{G_1}(e)$, $lab_{G_1^\prime}(e) = lab_{G_1}(e)$ for $e \in E_{G_1}$, and $att_{G_1^\prime}(e_0) = \varphi_1(1) \dotsc \varphi_1(k)$ (the label of $e_0$ does not matter). Let $G_2^\prime = \langle V_{G_2},E_{G_2},att_{G_2},lab_{G_2}, \varphi_2(1)\dotsc \varphi_2(k) \rangle$. Then $G_1 +_{\varphi_1,\varphi_2} G_2 = G_1^\prime[e_0/G_2^\prime]$.
\end{proposition}
This proposition immediately follows from the definitions of gluing and replacement.

A \emph{hypergraph grammar rule over a set of labels $C$} is of the form $r = (L \overset{\varphi_L}{\leftarrow} I \overset{\varphi_R}{\rightarrow} R)$ where $L,I,R \in \mathcal{H}(C)$ are zero rank and $\varphi_L$, $\varphi_R$ are morphisms. A hypergraph $G$ \emph{is transformed into $H$ via $r$} if there is a zero-rank hypergraph $C$ and a morphism $\psi: I \to C$ such that $G \cong C+_{\psi,\varphi_L} L$, $H \cong C+_{\psi,\varphi_R} R$ \cite{Konig18} ($\cong$ means that hypergraphs are isomorphic). Categorically, this can be expressed by a double pushout diagram:
\begin{center}
	\begin{tikzcd}
		L \arrow[d, "m"] & I \arrow[l, "\varphi_L"] \arrow[r, "\varphi_R"] \arrow[d, "\psi"] & R \arrow[d, "n"] \\
		G                & C \arrow[r, "\eta_R"] \arrow[l, "\eta_L"]                         & H               
	\end{tikzcd}
\end{center}
This transformation is denoted as $G \underset{r}{\Rightarrow} H$ or simply as $G \Rightarrow H$. 

NB! Hereinafter, we consider only hypergraph rules with $I$ being discrete. This does not substantially restrict the formalism.
\begin{example}\label{ex_dpo_rule}
	Consider the following DPO rule $\rho$:
	$$
	\rho = \left( \vcenter{\hbox{{\tikz[baseline=.1ex]{
					\node[node, label=above:{\scriptsize $\{1\}$}] (N) {};
					\node[node, below left=7mm and 7mm of N, label=below:{\scriptsize $\{ 2 \}$}] (N1) {};
					\node[node, below right=7mm and 7mm of N, label=below:{\scriptsize $\{ 3 \}$}] (N2) {};
					\draw[->,black] (N) -- node[left] {$l$} (N1);
					\draw[->,black] (N) -- node[right] {$r$} (N2);
	}}}}
	\;\leftarrow\; \Dis[3] \;\rightarrow\; 
	\vcenter{\hbox{{\tikz[baseline=.1ex]{
					\node[node, label=above:{\scriptsize $\{ 1 \}$}] (N) {};
					\node[hyperedge,below=4.74mm of N] (E1) {$t$};
					\node[node, below left=7mm and 7mm of N, label=below:{\scriptsize $\{ 2 \}$}] (N1) {};
					\node[node, below right=7mm and 7mm of N, label=below:{\scriptsize $\{ 3 \}$}] (N2) {};
					\node[hyperedge, right=5mm of N] (F) {$f$};
					\draw[-,black] (N) -- node[left] {\scriptsize 1} (E1);
					\draw[-,black] (N1) -- node[below] {\scriptsize 2} (E1);
					\draw[-,black] (N2) -- node[above] {\scriptsize 3} (E1);
					\draw[-,black] (N) -- node[above] {\scriptsize 1} (F);
	}}}} \right)
	$$
	Note that both the leftmost and the rightmost hypergraphs are zero rank (no external nodes in round brackets); numbers in curly brackets represent images of nodes of the interface hypergraph $\Dis[3]$ (i.e. $\varphi_L(1)$ is the node with the superscript $\{ 1 \}$ in the leftmost hypergraph and so on).
\end{example}

The definition of a hypergraph grammar rule and Proposition \ref{prop_dpo_to_replacement} imply that the rule application of $G \underset{r}{\Rightarrow} H$ for $r = (L \overset{\varphi_L}{\leftarrow} \Dis[k] \overset{\varphi_R}{\rightarrow} R)$ consists of an inverse replacement $G = C^\prime[e_0/L^\prime] \Leftarrow C^\prime$ and of a replacement $C^\prime \Rightarrow C^\prime[e_0/R^\prime] = H$. Here $C$ is as in the definition of a hypergraph rule application, and $C^\prime$, $e_0$, $L^\prime$ and $R^\prime$ are defined in the same way as $G_1^\prime$, $e_0$ and $G_2^\prime$ in Proposition \ref{prop_dpo_to_replacement}; in particular, $L^\prime = \langle V_{L},E_{L},att_{L},lab_{L}, \varphi_L(1)\dotsc \varphi_L(k) \rangle$, $R^\prime = \langle V_{R},E_{R},att_{R},lab_{R}, \varphi_R(1)\dotsc \varphi_R(k) \rangle$. 

An application of a hypergraph rule can be extended to cases where $G$ and $H$ are not zero rank. Indeed, we can say that $G$ is transformed into $H$ via $r$ if there is a hypergraph $C^\prime$ with a distinguished hyperedge (call it $e_0$) such that $C^\prime[e_0/L^\prime] \cong G$ and $C^\prime[e_0/R^\prime] \cong H$. Hereinafter we will use this extended definition. Clearly, if $G \Rightarrow H$, then $rk(G)=rk(H)$.

A \emph{DPO hypergraph grammar} $HGr$ is of the form $\langle N, \Sigma, P, Z\rangle$ where $N$, $\Sigma$ are disjoint finite alphabets of \emph{nonterminal} and \emph{terminal} labels resp., $P$ is a finite set of hypergraph grammar rules over $N \cup \Sigma$ of the form $L \overset{\varphi_L}{\leftarrow} \Dis[k] \overset{\varphi_R}{\rightarrow} R$, and $Z$ is a \emph{start hypergraph}. The language $L(HGr)$ generated by $HGr$ is the set of all hypergraphs $H \in \mathcal{H}(\Sigma)$ such that $Z \Rightarrow^\ast H$. Note that we can assume without loss of generality that $Z = S^\bullet$ for $S \in N$.

\subsection{Hypergraph Lambek Calculus and Hypergraph Lambek Grammars}\label{ssec_HL_HLG}
Now let us define the hypergraph Lambek calculus. Details concerning its motivation can be found in \cite{Pshenitsyn21_2, Pshenitsyn22}; we also provide a motivation relating the hypergraph Lambek calculus with linear logic in the next subsection. 

Let us fix a set $\mathit{Pr}$ of \emph{primitive types} along with a function $rk:\mathit{Pr} \to \mathbb{N}$; we require that for each $k \in\mathbb{N}$ there are countably many $p \in \mathit{Pr}$ such that $rk(p) = k$. Besides, we fix a countable set of labels $\$_n,n\in\mathbb{N}$ and set $rk(\$_n)=n$; let us agree that these labels do not belong to any other set considered in the definition of the calculus. Then the set of \emph{types} $\mathit{Tp}$ is defined inductively as follows:
\begin{enumerate}
	\item All primitive types are types.
	\item Let $N \in \mathit{Tp}$ be a type, and let $D$ be a hypergraph such that labels of all of its hyperedges, except for one, are from $\mathit{Tp}$, and one of them equals $\$_d$ for some $d$; let also $rk(N)=rk(D)$. Then $N\div D$ is also a type such that $rk(N\div D)\eqdef d$. The hyperedge of $D$ labeled by $\$_d$ is denoted by $e^\$_D$.
	\item If $M$ is a hypergraph labeled by types from $\mathit{Tp}$, then $\times(M)$ is also a type, and $rk(\times(M))\eqdef rk(M)$.
\end{enumerate}
Example \ref{ex_hl_type} contains an exemplar of a type (note that $\times$ binds stronger then $\div$). A \emph{sequent} is a structure of the form $H\to A$ where $H$ is a hypergraph labeled by types (called the \emph{antecedent} of the sequent), and $A$ is a type (called the \emph{succedent}) such that $rk(H)=rk(A)$.

The hypergraph Lambek calculus $\mathrm{HL}$ derives hypergraph sequents. The only axiom of $\HL$ is of the form $A^\bullet\to A$ where $A\in \mathit{Tp}$. There are four inference rules of $\HL$:
\begin{center}
	\begin{tabular}{cc}
		$
		\infer[(\div\to)]{H\left[e/D[e^\$_D/ (N\div D)^\bullet,d_1/H_1,\dotsc,d_k/H_k]\right]\to A}{H[e/N^\bullet]\to A & H_1\to lab_D(d_1) &\dots & H_k\to lab_D(d_k)}
		$
		&
		$
		\infer[(\to\div)]{F\to N\div D}{D[e^\$_D/F]\to N}
		$
		\\
		&
		\\
		$
		\infer[(\to\times)]{M[m_1/H_1,\dotsc,m_l/H_l]\to\times(M)}{H_1\to lab_M(m_1) & \dots & H_l\to lab_M(m_l)}
		$
		&
		$
		\infer[(\times\to)]{H[e/(\times(M))^\bullet]\to A}{H[e/M]\to A}
		$
		\\
	\end{tabular}
\end{center}
Here $N\div D$, $\times(M)$ are types; $e \in E_H$; $E_D=\{e^\$_D,d_1,\dots,d_k\}$, $E_M = \{m_1,\dotsc, m_l\}$. In each rule presented above, the sequents above the line are called \emph{premises}, and the sequent below the line is called the \emph{conclusion}. A hypergraph sequent $H\to A$ is said to be \emph{derivable in $\mathrm{HL}$} (denoted by $\mathrm{HL}\vdash H\to A$) if it can be obtained from axioms of $\HL$ by applications of rules of $\mathrm{HL}$. A corresponding sequence of rule applications is called a \emph{derivation}. An example of a derivation is given in Example \ref{ex_hl_derivation}.

Inference rules of the hypergraph Lambek calculus are defined as transformations operating on hypergraph sequents. All the rules are defined through replacement; besides, after an application of each rule a new type appears either in the antecedent or in the succedent of a sequent. Let us take a closer look at two particular rules, namely, at $(\div\to)$ and $(\times \to)$. The rule $(\div\to)$ is organized as follows: given a sequent $H[e/N^\bullet] \to A$ (note that $H[e/N^\bullet]$ is structurally the same hypergraph $H$, the replacement only changes the label of $e$) and sequents $H_i \to lab_D(d_i)$ for $i=1,\dotsc, k$, we replace $e$ in $H$ by $D$, then relabel the $\$_d$-labeled hyperedge by the type $(N \div D)$, and then replace each $d_i$ by the corresponding antecedent $H_i$ ($i=1,\dotsc,k$). Hence, this rule essentially consists of several replacements. In contrast, the rule $(\times\to)$ performs a transformation inverse to replacement: if one has a hypergraph $G[e/M]$ in the antecedent, then he/she can ``compress'' its subhypergraph $M$ into a single hyperedge $e$ labeled by the type $\times(M)$.

\begin{remark}\label{remark_invertibility}
	The rules $(\times\to)$ and $(\to\div)$ are invertible in $\HL$. This means that:
	\begin{enumerate}
		\item If a sequent of the form $H[e/(\times(M))^\bullet]\to A$ is derivable in $\HL$, then so is $H[e/M] \to A$.
		\item If a sequent of the form $F \to N \div D$ is derivable in $\HL$, then so is $D[e^\$_D/F]\to N$.
	\end{enumerate}
	Here all the notation is the same as for the rules $(\to \div)$ and $(\times\to)$. This can be proved using the cut elimination theorem; the theorem and its proof can be found in \cite[Proposition 1]{Pshenitsyn22}.
\end{remark}

An \emph{$\HL$-grammar} is a tuple $HLGr=\langle \Sigma, S, \triangleright\rangle$ where $\Sigma$ is an alphabet along with a rank function $rk:\Sigma \to \mathbb{N}$; $S\in \mathit{Tp}$ is a distinguished type; $\triangleright\subseteq\Sigma\times \mathit{Tp}$ is a finite binary relation such that $a\triangleright T$ implies $rk(a)=rk(T)$. \emph{The language $L(HLGr)$ generated by an $\HL$-grammar} $HLGr=\langle \Sigma, S, \triangleright\rangle$ is the set of all hypergraphs $G\in\mathcal{H}(\Sigma)$ for which a function $f_G:E_G\to \mathit{Tp}$ exists such that: 
\begin{enumerate}
	\item $lab_G(e)\triangleright f_G(e)$ whenever $e\in E_G$;
	\item $\mathrm{HL}\vdash f_G(G)\to S$ (recall that $f_G(G)$ is a relabeling of $G$ by means of $f_G$).
\end{enumerate}

\subsection{Some Insights Into HL}\label{ssec_insights}
In this subsection, we would like to outline the relationship between the Lambek calculus, $\HL$ and linear logic in order to provide some intuition for the hypergraph Lambek calculus. This can be done using ideas from \cite{MootP01}. In that paper, the authors introduce \emph{first-order multiplicative intuitionistic linear logic} $\MILLFO$ and show that the Lambek calculus can be embedded in it (as well as its variants like the Lambek calculus with permutation, the nonassociative Lambek calculus etc.). This calculus turns out to be closely related to $\HL$ as we are going to show. 

The content of this subsection is not necessary to understand the main technical results of the paper, so if the reader would like to skip it, we advise him/her to proceed with Section \ref{sec_hmel}.

The Lambek calculus is a propositional logic, which is usually presented in the Gentzen style, i.e. as a sequent calculus. The use of sequent calculi is convenient, because it is easier to check derivability using them than using a Hilbert-style axiom-centered calculus. Besides, if we consider e.g. the classical propositional calculus in the Gentzen style, then we have a nice division of its rules into two parts: there are purely logical rules describing behaviour of logical operations like conjunction, disjunction, impication, or negation, and there are structural rules like weakening, contraction, or permutation. For example, in classical logic the rule of weakening is of the following form:
$$
\infer[]{
	\Gamma, B \to A
}{
	\Gamma \to A
} 
$$
It is understood as follows: if one can prove $A$ from assumptions $\Gamma$, then he/she can also prove it from $\Gamma,B$. Clearly, such a rule is not concerned with a particular logical operation but it tells us something about the sequent behaviour itself. By dropping some or all the structural rules, one obtains \textit{substructural logics}, which turn out to be useful in computer science, linguistics and other branches of science. The Lambek calculus is one of them; it includes neither weakening nor contraction nor even permutation. Consequently, the order of assumptions in the left-hand side of a sequent matters. This gives rise to a non-commutative version of conjunction called the product and to two its residues called the left division and the right division. Formally, types (formulas) of $\LC$ are built from primitive ones $\mathit{Pr}$ using these three operations denoted as $\cdot$, $\backslash$, and $/$ resp. A sequent is a structure of the form $A_1,\dotsc,A_n\to A$ where $n > 0$ and $A_i$, $A$ are types. The only axiom of $\LC$ is of the form $A\to A$ for each type $A$. There are six rules:
$$
\infer[(\cdot\to)]{\Gamma, A \cdot B, \Delta \to C}{\Gamma, A, B, \Delta \to C}
\qquad
\infer[(\to\cdot)]{\Pi, \Psi \to A \cdot B}{\Pi \to A & \Psi \to B}
$$
$$
\infer[(\backslash\to)]{\Gamma, \Pi, A \backslash B, \Delta \to C}{\Pi \to A & \Gamma, B, \Delta \to C}
\quad\;
\infer[(\to\backslash)]{\Pi \to A \backslash B}{A, \Pi \to B}
\quad\;
\infer[(/\to)]{\Gamma, B / A, \Pi, \Delta \to C}{\Pi \to A & \Gamma, B, \Delta \to C}
\quad\;
\infer[(\to/)]{\Pi \to B / A}{\Pi, A \to B}
$$
Here capital Latin letters denote types, and capital Greek letters denote sequences of types (and $\Pi$ is nonempty).
\begin{example}\label{ex_sequent_LC}
	The sequent $p, q/r, r/s \to (p\cdot q)/s$ is derivable in $\LC$ (where $p$, $q$, $r$, and $s$ are primitive types):
	$$
	\infer[(\to /)]{
		p, q/r, r/s \to (p\cdot q)/s
	}
	{
		\infer[(/ \to)]{
			p, q/r, r/s, s \to p\cdot q
		}
		{
			\infer[(\to \cdot)]{
				p, q \to p\cdot q
			}
			{
				p \to p & q \to q
			}
			&
			\infer[(/ \to)]{
				r/s, s \to r
			}
			{
				r \to r & s \to s
			}
		}
	}
	$$
\end{example} 

There are many modifications of the Lambek calculus designed for linguistic and logical purposes: the multimodal Lambek calculus, the displacement calculus, the nonassociative Lambek calculus etc. Besides, there are calculi generalizing and unifying these modifications, i.e. they represent the latter in a uniform setting. One of such generalizations is the hypergraph Lambek calculus. Another one is first-order intuitionistic linear logic \cite{MootP01}, or, more precisely, its multiplicative fragment $\MILLFO$. Its language includes individual variables $x_0,x_1,\dotsc$, individual constants $c_0,c_1,\dotsc$, functional symbols of different arities, the binary connectives $\otimes$ and $\multimap$, and the quantifiers $\forall$ and $\exists$. A \textit{term} is an application of a functional symbol to a list of variables and constants. \textit{Formulas} of $\MILLFO$ are built from terms using $\otimes$, $\multimap$, and the quantifiers as in any first-order logic. For example, $\forall x.\exists y. (s(x,y) \otimes t(x,y,c_0))$ is a formula. A \textit{sequent} is of the form $\Gamma \to A$ where $\Gamma$ is a \textit{multiset} of formulas (i.e. it is not ordered) and $A$ is a formula. The only axiom scheme of $\MILLFO$ is $A \to A$ for all formulas $A$. The inference rules are as follows:
$$
\infer[(\otimes \to)]{
	\Gamma, A\otimes B \to C
}{
	\Gamma, A, B \to C
}
\quad
\infer[(\to \otimes)]{
	\Gamma, \Delta \to A\otimes B 
}{
	\Gamma \to A
	&
	\Delta \to B
}
\quad
\infer[(\multimap\to)]{
	\Gamma, \Delta, A\multimap B \to C
}{
	\Delta \to A
	&
	\Gamma, B \to C
}
\quad
\infer[(\to \multimap)]{
	\Gamma \to A\multimap B
}{
	\Gamma, A \to B
}
$$
$$
\infer[(\exists \to)]{
	\Gamma, \exists x.A \to C
}{
	\Gamma, A \to C
}
\quad
\infer[(\to \exists)]{
	\Gamma \to \exists x.A
}{
	\Gamma \to A[x\eqdef e]
}
\quad
\infer[(\forall \to)]{
	\Gamma, \forall x.A \to C
}{
	\Gamma, A[x:=e] \to C
}
\quad
\infer[(\to \forall)]{
	\Gamma \to \forall x.A
}{
	\Gamma \to A
}
$$
Here $e$ is an arbitrary constant or variable of our choice (either present in $\Gamma$ or in $A$ or a fresh one). It is required that in the rules $(\exists \to)$ and $(\to \forall)$ the variable $x$ does not occur freely in $\Gamma$ or $C$.

In \cite{MootP01}, the translation of $\LC$ into $\MILLFO$ is discovered such that each type of $\LC$ is transformed into a formula of $\MILLFO$ with two free variables:
\begin{enumerate}
	\item $||p||^{x,y} = p(x,y)$;
	\item $||A/B||^{x,y} = \forall z. ||B||^{y,z} \multimap ||A||^{x,z}$;
	\item $||B \backslash A||^{x,y} = \forall z. ||B||^{z,x} \multimap ||A||^{z,y}$;
	\item $||A \cdot B||^{x,y} = \exists z. ||A||^{x,z} \otimes ||B||^{z,y}$.
\end{enumerate}
Finally, a sequent $A_1,\dotsc,A_n \to B$ of $\LC$ is translated into the sequent $||A_1||^{c_0,c_1},\dotsc,||A_n||^{c_{n-1},c_n} \to ||B||^{c_0,c_n}$ where $c_0,c_1,\dotsc,c_n$ are distinct constants. It is proved \cite{MootP01} that a sequent is derivable in $\LC$ if and only if its translation is derivable in $\MILLFO$.

\begin{example}\label{ex_translation_MILL1}
	The types from the sequent presented in Example \ref{ex_sequent_LC} are translated into $\MILLFO$ as follows:
	\begin{enumerate}
		\item $||p||^{x,y} = p(x,y)$;
		\item $||q/r||^{x,y} = \forall z. r(y,z) \multimap q(x,z)$;
		\item $||r/s||^{x,y} = \forall z. s(y,z) \multimap r(x,z)$;
		\item $||(p\cdot q)/s||^{x,y} = \forall z. (s(y,z) \multimap \exists t. p(x,t)\otimes q(t,z))$.
	\end{enumerate}
	The sequent from Example \ref{ex_sequent_LC} is translated as follows:
	\begin{equation}\label{eq_ex_translation_MILL1}
		p(c_0,c_1), \forall z. r(c_2,z) \multimap q(c_1,z), \forall z. s(c_3,z) \multimap r(c_2,z) \;\to\; \forall z. (s(c_3,z) \multimap \exists t. p(c_0,t)\otimes q(t,z))
	\end{equation}
	It is an exercise to verify that the latter sequent is derivable in $\MILLFO$.
\end{example}
An important observation is that, although $\MILLFO$ is a commutative logic, that is, the order of formulas in antecedents of sequents can be freely changed, we can still embed a noncommutative logic like $\LC$ in it. The trick is that we preserve the order of types in $\LC$ by using constants $c_0,\dotsc,c_n$: they fix the linear structure of an original sequent from $\LC$.

In \cite{Pshenitsyn22}, we prove that $\LC$ can be embedded in $\HL$. E.g. the sequent $p,q/r, r/s \to (p\cdot q)/s$ is translated into the following hypergraph sequent:
\begin{equation}\label{eq_seq_SG}
	\mbox{	
		{\tikz[baseline=.1ex]{
				\node[node,label=left:{\scriptsize $(1)$}] (N1) {};
				\node[node,right=20mm of N1] (N2) {};
				\node[node,right=20mm of N2] (N3) {};
				\node[node,right=20mm of N3,label=right:{\scriptsize $(2)$}] (N4) {};
				\draw[->,black] (N1) -- node[above] {$p$} (N2);
				\draw[->,black] (N2) -- node[above] {$tr_{\LC}(q/r)$} (N3);
				\draw[->,black] (N3) -- node[above] {$tr_{\LC}(r/s)$} (N4);
	}}}\to \; tr_{\LC}((p\cdot q)/s)
\end{equation}
Here $tr_{\LC}$ is a translation from $\LC$ into $\HL$, which we are not going to describe in this paper. Comparing the embeddings of $\LC$ to $\MILLFO$ and to $\HL$ we observe a number of correspondences: constants $c_i$ ($i=0,1,2,3$) correspond to the nodes of the graph in the antecedent of (\ref{eq_seq_SG}); the formula $p(c_0,c_1)$ corresponds to a $p$-labeled edge going from the node $c_0$ to the node $c_1$; constants $c_0$, $c_2$ correspond to the external nodes of the graph. We can also compare the rules of $\HL$ and the rules of $\MILLFO$ and observe strong similarity. In fact, these observations can be generalized: we claim that the hypergraph Lambek calculus can be embedded in $\MILLFO$ (see an example below). However, formally introducing this embedding and proving its correctness should be a matter of another paper; moreover, this is not relevant for our further considerations in this work. Recall that the main goal is to provide an intuition concerning how $\HL$ is organized and how to add the exponential modality to $\HL$.
\begin{example}
	The type $\DPO(\rho)$ from Example \ref{ex_hl_type} can be translated into the following formula of $\MILLFO$: $\forall x.\forall y.\forall z. (f(x)\otimes t(x,y,z)) \multimap (l(x,y)\otimes r(x,z))$. The sequent of $\HL$, which appears at the end of the derivation from Example \ref{ex_hl_derivation}, can be translated into the following $\MILLFO$ sequent:
	\begin{eqnarray*}
		f(c_0), t(c_0,c_1,c_2), p(c_1,c_2), \forall x.\forall y.\forall z. (f(x)\otimes t(x,y,z)) \multimap (l(x,y)\otimes r(x,z)) \\ \;\to\; \exists x. \exists y. \exists z. l(x,y) \otimes r(x,z) \otimes p(y,z)
	\end{eqnarray*}
\end{example}
Having this in mind, we would like to proceed with extending $\HL$ by the exponential.

\section{The Hypergraph Multiplicative-Exponential Lambek Calculus}\label{sec_hmel}

In this section, we show a way of extending $\HL$ with the exponential modality $!$ resulting in the \emph{hypergraph multiplicative exponential Lambek calculus} $\HMEL$ (we will explain this $0$ subscript later). After doing this we show how to convert any DPO grammar into an equivalent $\HMEL$-grammar. The conversion procedure is similar to that presented in \cite{Kanazawa99}: there, each unrestricted Chomsky grammar is converted into an equivalent grammar over the multiplicative-exponential Lambek calculus. 

In Section \ref{ssec_insights}, we introduced the rules for the fragment of intuitionistic linear logic. Now, let us look at the full logic $\ILLFO$, which includes exponentials. Formulas of $\ILLFO$ are built using $\otimes$, $\multimap$, $\mathbf{1}$, $\&$, $\oplus$, $!$ and the quantifiers $\forall$ and $\exists$. The rules of $\ILLFO$ include those of $\MILLFO$ and those defining the new connectives. Let us focus on the rules for the $!$ modality:
$$
\infer[(!\to)]{\Gamma,!A \to B}{\Gamma,A \to B}
\qquad
\infer[(\to!)]{!A_1,\dotsc,!A_n \to !B}{!A_1,\dotsc,!A_n \to B}
\qquad
\infer[(\mathrm{w})]{\Gamma,!A \to B}{\Gamma \to B}
\qquad
\infer[(\mathrm{c})]{\Gamma,!A \to B}{\Gamma,!A,!A \to B}
$$
Our aim is to transfer these rules to $\HL$. It would be great to introduce the exponential modality added to $\HL$ unrestrictedly, that is, to be able to consider the type $!A$ for each type $A$. However, this faces certain difficulties if we expect the resulting calculus to be well-behaving, namely, to enjoy the cut elimination theorem. Unfortunately, for now we have not invented a general treatment of exponentials in the hypergraph calculus in such a way that this theorem holds, so this remains an open question. However, if we restrict ourselves and allow one to consider the type $!A$ only if $A$ has the rank $0$ ($rk(A)=0$), then it turns out that such a calculus can be well defined; in particular, it satisfies the cut elimination theorem.

The rules for the exponential modality we suggest to add to $\HL$ are the following ones:
\begin{equation}\label{eq_rules_HMEL}
	\infer[(!\to)]{H+(!A)^\bullet \to B}{H+A^\bullet \to B}
	\qquad
	\infer[(\to!)]{\sum_{i=1}^{n} (!A_i)^\bullet \to !A}{\sum_{i=1}^{n} (!A_i)^\bullet \to A}
	\qquad
	\infer[(\mathrm{w})]{H+(!A)^\bullet \to B}{H \to B}
	\qquad
	\infer[(\mathrm{c})]{H+(!A)^\bullet \to B}{H+2\cdot (!A)^\bullet \to B}
\end{equation}
Here $rk(!A_i)=rk(A_i)=rk(!A)=rk(A)-0$. The summation symbol stands for multiple disjoint union. Clearly, these rules generalize their string counterparts considered earlier for $\ILLFO$. Note that if $rk(X)=0$, then $X^\bullet$ is a zero-rank hyperedge ``floating'' as a separate component of an antecedent hypergraph.

\begin{definition}\label{def_HMEL}
	The \emph{hypergraph multiplicative-exponential Lambek calculus} $\HMEL$ is defined similarly to $\HL$ but in the definition of types $\mathit{Tp}$ we additionally say that, if $A$ is a type such that $rk(A)=0$, then $!A$ is a type as well such that $rk(!A)=0$. We add the rules from (\ref{eq_rules_HMEL}) to the rules of $\HL$.
	\\
	$\HMEL$-grammars are defined in the same way as $\HL$-grammars but based on $\HMEL$.
\end{definition}
The subscript $0$ corresponds to the restriction on putting $!$ only on types with rank $0$. As we mentioned, it allows us to prove the following theorem:

\begin{theorem}\label{th_cut_HMEL}
	If $\HMEL \vdash H \to A$ and $\HMEL \vdash G[e_0/A^\bullet] \to B$, then $\HMEL \vdash G[e_0/H] \to B$.
\end{theorem}
In other words, this theorem states that we can add the cut rule introduced below to the list of rules without enlarging the set of derivable sequents:
$$
\infer[(\mathrm{cut})]{G[e_0/H] \to B}{H \to A & G[e_0/A^\bullet] \to B}
$$
\begin{proof}
	We prove that both the cut rule and the \textit{mix rule}, which we introduce below, can be added to $\HMEL$ without affecting the set of derivable sequents:
	$$
	\infer[(\mathrm{mix})]{G^\prime+H \to B}{H \to !C & G^\prime+(!C)\cdot n \to B}
	$$
	Here $H$ is zero rank. Such a rule is commonly introduced to prove the cut elimination theorem for variants of linear logic including weakening (w) and contraction (c). In fact, the whole proof resembles that from \cite{LincolnMSS92}, and it is typical for substructural logics with exponentials like $\HMEL$. Let us denote $!C$ by $A$ and both $G^\prime+(!C)\cdot n$ and $G[e_0/A^\bullet]$ by $F$ for the sake of uniformity.
	
	The proof is done by nested induction: the outer one is on the size of $A$ (counted as the total number of primitive types and of symbols $\times$, $\div$ and $!$ in the construction of $A$), and the inner one is on the sum of lengths of the derivations of $H \to A$ and $F \to B$. We need to consider several cases depending on the last rule applied in the derivation of $H \to A$ or in that of $F \to B$. An important notion is that of \textit{the major type}: given a derivation of some sequent, the major type is the type that appears in the sequent after the last step of the derivation. E.g., in Example \ref{ex_sequent_LC}, the major type is $(p \cdot q)/s$, and in Example \ref{ex_hl_derivation}, it is $\DPO(\rho)$.
	
	The first group of cases is where $H \to A$ or $F \to B$ is an axiom; then the statement of the theorem becomes trivial. 
	
	The second group is where $A$ from the succedent of $H \to A$ is not major in the derivation of $H \to A$. Then the last rule is concerned with $H$ somehow, and it does not affect $A$. Thus we can repeat the same step but when $H$ replaces $e$ within $G$, and after that apply the induction hypothesis. More formally, let the last step of the derivation of $H \to A$ be as follows:
	$$
	\infer[]{
		H \to A
	}{
		H^\prime \to A
		&
		\mathcal{P}
	} 
	$$
	Here $\mathcal{P}$ are some additional premises. Then we apply the induction hypothesis to $H^\prime \to A$ and $F \to B$ and thus conclude that $G[e_0/H^\prime] \to B$ ($G^\prime+H^\prime \to B$ for the mix rule) is derivable. Finally, we perform the following step:
	$$
	\infer[]{
		G[e_0/H] \to B
	}{
		G[e_0/H^\prime] \to B
		&
		\mathcal{P}
	}
	$$
	This completes the proof for this group of cases. From now on, we assume that $A$ is major in $H \to A$.
	
	The third group of cases is where the distinguished occurrence of $A$ (or all the $n$ distinguished occurrences of $A=!C$ if we consider $(\mathrm{mix})$) is not major in the derivation of $F \to B$ and the last rule in its derivation is not $(\to !)$. Then similar reasonings to those for the second group can be applied.
	
	The fourth group is where $A$ is major in both $H \to A$ and $F \to B$ but $A$ is not of the form $!C$ (it is possible only for the cut rule). Then the reasonings are the same as for $\HL$, see \cite[Appendix A]{Pshenitsyn22}.
	
	In the remaining cases $A=!C$ (thus the cut rule becomes an instance of the mix rule) is major in both $H \to A$ and $F \to B$; we also need to cover the case where the last rule in the derivation of $F \to B$ is $(\to !)$. In all these cases, the last rule in the derivation of $H \to A$ is $(\to !)$, consequently, the sequent $H \to C$ is derivable and $H=\sum_{i=1}^n (!A_i)^\bullet$ for some $A_i$ such that $rk(A_i)=0$. The following cases are possible depending on the last rule applied in the derivation of $F \to B$:
	
	\textit{Case I.} The last rule application is $(!\to)$. Then it must be of the form:
	$$
	\infer[(! \to)]{G^\prime+(!C)\cdot n \to B}{G^\prime+(!C)\cdot (n-1) + C^\bullet \to B}
	$$
	Using the induction hypothesis, we apply the mix rule to $H \to A$ and $G^\prime+A\cdot (n-1) + C^\bullet \to B$ concluding that $G^\prime+H+C^\bullet \to B$ is derivable. Then we again apply the induction hypothesis to the sequents $H \to C$ and $G^\prime+H+C^\bullet \to B$ (note that the size of $C$ is less than that of $A$ so the first induction parameter decreases) coming up with the sequent $G^\prime + H + H \to B$. Finally, one applies the contraction rule $(\mathrm{c})$ $n$ times thus contracting two $H$'s into a single $H$.
	
	\textit{Case II.} The last rule application is $(\mathrm{w})$. Then it must be of the form:
	$$
	\infer[(! \to)]{G^\prime+(!C)\cdot n \to B}{G^\prime+(!C)\cdot (n-1) \to B}
	$$
	It remains to apply the induction hypothesis to $H \to A$ and $G^\prime+(!C)\cdot (n-1) \to B$.
	
	\textit{Case III.} The last rule application is $(\mathrm{c})$. The reasonings are similar to the previous case.
	
	\textit{Case IV.} The last rule application is $(\to !)$. In what follows, $G^\prime = \sum_{i=1}^m (!D_i)^\bullet$ for some types $D_i$, $B=!D$ and $F \to D$ is derivable. Let us apply the induction hypothesis to $H \to A$ and $F \to D$ thus concluding that $\sum_{i=1}^m (!B_i)^\bullet+\sum_{i=1}^n (!A_i)^\bullet \to D$ is derivable. Finally, we apply the rule $(\to !)$, which results in $G^\prime + H \to B$ as desired.
\end{proof}
As a consequence of the cut elimination theorem, we notice that Remark \ref{remark_invertibility} holds for $\HMEL$ as well.

Now our aim is to prove that $\HMEL$-grammars based on the new calculus are at least as expressive as $\DPO$ grammars. This can be done similarly to the proof of the fact that grammars over the multiplicative-exponential Lambek calculus are at least as expressive as unrestricted Chomsky grammars. Recall that a rule in a Chomsky grammar is of the form $\alpha \Rightarrow \beta$ where $\alpha,\beta$ are two arbitrary strings of nonterminal and terminal symbols. The proof from \cite{Kanazawa99} suggests converting each rule of the form $A_1\dotsc A_n \Rightarrow B_1 \dotsc B_m$ (where $A_i,B_j$ are nonterminal symbols) into the following type: 
$$[\tau(A_1\dotsc A_n \Rightarrow B_1 \dotsc B_m)] \eqdef (A_1\cdot \dotsc \cdot A_n)/(B_1\cdot \dotsc \cdot B_m).$$
Makoto Kanazawa proves in \cite{Kanazawa99} that, given a finite set $R = \{r_1,\dotsc,r_k\}$ of rules of the above form, a string $w$ of nonterminal symbols is derivable from the nonterminal symbol $S$ using rules from $R$ if and only if $\mathrm{MELC} \vdash ![\tau(r_1)],\dotsc,![\tau(r_k)], w \to S$ where all the nonterminal symbols are considered as primitive types and where all symbols in $w$ are separated by commas. E.g. the rule $S \Rightarrow SSA$ allows one to produce $SSASA$ from $S$, hence $\mathrm{MELC} \vdash !(S/(S\cdot S \cdot A)), S, S, A, S, A \to S$. $\mathrm{MELC}$ stands for the multiplicative-exponential Lambek calculus.

The same idea can be implemented for $\DPO$ rules. Firstly we need to slightly enhance DPO hypergraph grammars.
\begin{construction}\label{construction_normalized}
	Given a DPO grammar $\langle N, \Sigma, P, S^\bullet\rangle$, we convert it into a grammar $\langle N^\prime, \Sigma, P^\prime, S^\bullet\rangle$, which we call \emph{normalized}, as follows. For each $a \in \Sigma$ we introduce a new nonterminal label $T_a$ with $rk(T_a)=rk(a)$; let $N^\prime = N \sqcup \{T_a \mid a \in \Sigma\}$. Then for each $r = (L \overset{\varphi_L}{\leftarrow} \Dis[k] \overset{\varphi_R}{\rightarrow} R) \in P$ we replace each terminal label $a$ in $L$, $R$ by $T_a$. Let us call such new rules \emph{nonterminal} and denote the set of nonterminal rules as $P_N$. Finally, we add rules that allow one to replace $T_a$ by $a$, i.e., rules of the form $(T_a^\circ \overset{\varphi_L}{\leftarrow} \Dis[k] \overset{\varphi_R}{\rightarrow} a^\circ)$ where $rk(a)=k$, $\varphi_L(i)=\varphi_R(i) = i$ for $i=1,\dotsc,k$ (here we use the notation of nodes as in the definitions of $\Dis[k]$ and $S^\circ$). These rules are called \emph{terminal} and are denoted as $P_T$. Finally, $P^\prime = P_N \cup P_T$.
	
	Clearly, the normalized grammar generates the same language as the original one.
	Hereinafter we consider only normalized grammars.
\end{construction}

Nonterminal rules can be converted into corresponding types of $\HL$ as follows:
\begin{construction}\label{construction_dpo_rule}
	Let us consider nonterminal labels of normalized grammars as primitive types (with the same rank function). If $r = (L \overset{\varphi_L}{\leftarrow} \Dis[k] \overset{\varphi_R}{\rightarrow} R)$ is a nonterminal rule, then $\DPO(r) \eqdef \times\left(\widehat{L}\right) \div \left(\widehat{R} + \$_0^\bullet\right)$ where $\widehat{L} = \langle V_L,E_L, att_L, lab_L, \varphi_L(1)\dotsc \varphi_L(k)\rangle$; $\widehat{R} = \langle V_R,E_R, att_R, lab_R, \varphi_R(1)\dotsc \varphi_R(k)\rangle$.
	\\
	Note that $\$_0^\bullet$ is a separate hyperedge of rank $0$ and that $rk(\DPO(r))=0$.
\end{construction}

\begin{example}\label{ex_hl_type}
	The nonterminal rule $\rho$ from Example \ref{ex_dpo_rule} is converted into the following type $\DPO(\rho)$:
	$$
	\DPO(\rho) = \times\left( \vcenter{\hbox{{\tikz[baseline=.1ex]{
					\node[node, label=below:{\scriptsize $(1)$}] (N) {};
					\node[node, below left=4mm and 10mm of N, label=below:{\scriptsize $(2)$}] (N1) {};
					\node[node, below right=4mm and 10mm of N, label=below:{\scriptsize $(3)$}] (N2) {};
					\draw[->,black] (N) -- node[above left] {$l$} (N1);
					\draw[->,black] (N) -- node[above right] {$r$} (N2);
	}}}} \right) \div \left( \vcenter{\hbox{{\tikz[baseline=.1ex]{
					\node[node, label=left:{\scriptsize $(1)$}] (N) {};
					\node[hyperedge,below=4.74mm of N] (E1) {$t$};
					\node[node, below left=7mm and 7mm of N, label=left:{\scriptsize $(2)$}] (N1) {};
					\node[node, below right=7mm and 7mm of N, label=right:{\scriptsize $(3)$}] (N2) {};
					\node[hyperedge, right=5mm of N] (F) {$f$};
					\draw[-,black] (N) -- node[right] {\scriptsize 1} (E1);
					\draw[-,black] (N1) -- node[below] {\scriptsize 2} (E1);
					\draw[-,black] (N2) -- node[above] {\scriptsize 3} (E1);
					\draw[-,black] (N) -- node[above] {\scriptsize 1} (F);
	}}}}\quad\vcenter{\hbox{{\tikz[baseline=.1ex]{\node[hyperedge] {$\$_0$};}}}} \right)
	$$
\end{example}

The main connection between $r$ and $\DPO(r)$ is shown in 
\begin{lemma}\label{lemma_dpo}
	Let $Y,Y^\prime$ be two hypergraphs and let $Y\underset{r}{\Rightarrow}Y^\prime$ for $r\in P_N$. If $\mathrm{HL} \vdash Y \to A$ for any type $A$, then $\mathrm{HL} \vdash Y^\prime+\DPO(r)^\bullet \to A$ as well.
\end{lemma}
\begin{lemma}\label{lemma_dpo_many}
	Let $Y\Rightarrow^\ast Y^\prime$ in a normalized grammar $\langle N, \Sigma, P, S^\bullet\rangle$ by means of nonterminal rules where $Y,Y^\prime$ are hypergraphs. If $\mathrm{HL} \vdash Y \to A$ for any type $A$, then $\mathrm{HL} \vdash Y^\prime+\sum\limits_{r \in P_N} k_{r}\cdot \DPO(r)^\bullet \to A$ for some $k_r \in \mathbb{N}$.
\end{lemma}
\begin{lemma}[main]\label{lemma_dpo_main}
	Let $Y^\prime$ be a hypergraph; let $\langle N, \Sigma, P, S^\bullet \rangle$ be a normalized grammar with the set $P_N$ of nonterminal rules; let $X$ be a nonterminal symbol. Then $X^\bullet \Rightarrow^\ast Y^\prime$ using rules from $P_N$ if and only if $\mathrm{HL} \vdash Y^\prime+\sum\limits_{r \in P_N} b_r\cdot (!\DPO(r))^\bullet + \sum\limits_{r \in P_N} k_r\cdot (\DPO(r))^\bullet \to X$ for some $b_r,k_r \in \mathbb{N}$.
\end{lemma}
Lemma \ref{lemma_dpo} is proved by straightforwardly applying $(\times\to)$ and then $(\div\to)$ to the sequent $Y \to A$ (below a representative example is provided). Lemma \ref{lemma_dpo_many} directly follows from Lemma \ref{lemma_dpo}. Note that $k_r$ in it is the number of applications of $r$ in the derivation $Y\Rightarrow^\ast Y^\prime$.

Therefore, a DPO derivation can be remodeled within $\HL$ but each rule application of $r$ leaves a trace, namely, a floating hyperedge labeled by $\DPO(r)$ in the antecedent. Using the exponential modality we can unify these hyperedges into a single one.

\begin{example}\label{ex_hl_derivation}
	The following derivation illustrates Lemma \ref{lemma_dpo}:
	$$
	\infer[(\div\to)]{
		\vcenter{\hbox{{\tikz[baseline=.1ex]{
						\node[node] (N) {};
						\node[hyperedge,below=3.5mm of N] (E1) {$t$};
						\node[node, below left=11mm and 6.3mm of N] (N1) {};
						\node[node, below right=11mm and 6.3mm of N] (N2) {};
						\node[hyperedge, right=5mm of N] (F) {$f$};
						\draw[->,black] (N1) -- node[below] {$p$} (N2);
						\draw[-,black] (N) -- node[left] {\scriptsize 1} (E1);
						\draw[-,black] (N1) -- node[above] {\scriptsize 2} (E1);
						\draw[-,black] (N2) -- node[above] {\scriptsize 3} (E1);
						\draw[-,black] (N) -- node[above] {\scriptsize 1} (F);
		}}}}\;\;\vcenter{\hbox{{\tikz[baseline=.1ex]{\node[hyperedge] {$\DPO(\rho)$};}}}}\quad \to \quad
		\times\left(\vcenter{\hbox{{\tikz[baseline=.1ex]{
						\node[node] (N) {};
						\node[node, below left=6.3mm and 6.3mm of N] (N1) {};
						\node[node, below right=6.3mm and 6.3mm of N] (N2) {};
						\draw[->,black] (N) -- node[above left] {$l$} (N1);
						\draw[->,black] (N) -- node[above right] {$r$} (N2);
						\draw[->,black] (N1) -- node[below] {$p$} (N2);
		}}}}\right)
	}
	{
		\infer[(\times\to)]{
			\vcenter{\hbox{{\tikz[baseline=.1ex]{
							\node[node] (N) {};
							\node[hyperedge,below=3.5mm of N] (E1) {$\times(\widehat{L})$};
							\node[node, below left=11mm and 6.3mm of N] (N1) {};
							\node[node, below right=11mm and 6.3mm of N] (N2) {};
							\draw[->,black] (N1) -- node[below] {$p$} (N2);
							\draw[-,black] (N) -- node[left] {\scriptsize 1} (E1);
							\draw[-,black] (N1) -- node[above] {\scriptsize 2} (E1);
							\draw[-,black] (N2) -- node[above] {\scriptsize 3} (E1);
			}}}} \to \times\left(\vcenter{\hbox{{\tikz[baseline=.1ex]{
							\node[node] (N) {};
							\node[node, below left=6.3mm and 6.3mm of N] (N1) {};
							\node[node, below right=6.3mm and 6.3mm of N] (N2) {};
							\draw[->,black] (N) -- node[above left] {$l$} (N1);
							\draw[->,black] (N) -- node[above right] {$r$} (N2);
							\draw[->,black] (N1) -- node[below] {$p$} (N2);
			}}}}\right)
		}
		{
			\infer[(\to\times)]{
				\vcenter{\hbox{{\tikz[baseline=.1ex]{
								\node[node] (N) {};
								\node[node, below left=6.3mm and 6.3mm of N] (N1) {};
								\node[node, below right=6.3mm and 6.3mm of N] (N2) {};
								\draw[->,black] (N) -- node[above left] {$l$} (N1);
								\draw[->,black] (N) -- node[above right] {$r$} (N2);
								\draw[->,black] (N1) -- node[below] {$p$} (N2);
				}}}} \to \times\left(\vcenter{\hbox{{\tikz[baseline=.1ex]{
								\node[node] (N) {};
								\node[node, below left=6.3mm and 6.3mm of N] (N1) {};
								\node[node, below right=6.3mm and 6.3mm of N] (N2) {};
								\draw[->,black] (N) -- node[above left] {$l$} (N1);
								\draw[->,black] (N) -- node[above right] {$r$} (N2);
								\draw[->,black] (N1) -- node[below] {$p$} (N2);
				}}}}\right)
			}
			{
				l^\bullet \to l & r^\bullet \to r & p^\bullet \to p
			}
		}
		&
		t^\bullet \to t
		&
		f^\bullet \to f
	}
	$$
	Here the sequent $Y \to A$ equals $\vcenter{\hbox{{\tikz[baseline=.1ex]{
					\node[node] (N) {};
					\node[node, below left=3mm and 3mm of N] (N1) {};
					\node[node, below right=3mm and 3mm of N] (N2) {};
					\draw[->,black] (N) -- node[above left] {$l$} (N1);
					\draw[->,black] (N) -- node[above right] {$r$} (N2);
					\draw[->,black] (N1) -- node[below] {$p$} (N2);
	}}}} \to \times\left(\vcenter{\hbox{{\tikz[baseline=.1ex]{
					\node[node] (N) {};
					\node[node, below left=3mm and 3mm of N] (N1) {};
					\node[node, below right=3mm and 3mm of N] (N2) {};
					\draw[->,black] (N) -- node[above left] {$l$} (N1);
					\draw[->,black] (N) -- node[above right] {$r$} (N2);
					\draw[->,black] (N1) -- node[below] {$p$} (N2);
	}}}}\right)$, and the rule $r$ equals $\rho$ from Example \ref{ex_dpo_rule}.
\end{example}
It remains to prove Lemma \ref{lemma_dpo_main}:
\begin{proof}[Proof (of Lemma \ref{lemma_dpo_main}).]
	The ``only if'' direction straightforwardly follows from Lemma \ref{lemma_dpo_many}: if we take $Y=X^\bullet$, $A=X$, then $Y \to A$ is an axiom $X^\bullet \to X$, hence $Y^\prime+\sum\limits_{r \in P_N} k_{r}\cdot \DPO(r)^\bullet \to A$ is derivable for some $k_r$. Finally, one can derive the sequent $Y^\prime+\sum\limits_{r \in P_N} (!\DPO(r))^\bullet \to X$ from it using the rules $(!\to)$, $(\mathrm{w})$, and $(\mathrm{c})$. 
	
	The ``if'' direction is proved by induction on the length of a derivation of $Y^\prime+\sum\limits_{r \in P_N} b_r\cdot (!\DPO(r))^\bullet+\sum\limits_{r \in P_N} k_r\cdot (\DPO(r))^\bullet \to X$ in $\HMEL$. For the sake of brevity, we denote the sum $\sum\limits_{r \in P_N} b_r\cdot (!\DPO(r))^\bullet+\sum\limits_{r \in P_N} k_r\cdot (\DPO(r))^\bullet$ as $\sigma\{b_r;k_r\}$.
	
	The base case is trivial. Let us prove the induction step. The proof is done by considering the last rule applied in a derivation of a sequent.
	
	\textit{Case 1.} The last rule applied is $(\div \to)$. In this case, its application must be of the form
	$$
	\infer[]{
		Y^\prime+\sigma\{b^\prime_r,k^\prime_r\}+ (\DPO(r_0))^\bullet \to X
	}{
		Y^0\left[e_0/(\times(\widehat{L}))^\bullet\right]+\sigma\{b^0_r,k^0_r\} \to X
		&
		Y^1+\sigma\{b^1_r,k^1_r\} \to lab_{\widehat{R}}(e_1)
		&
		\dotsc
		&
		Y^l+\sigma\{b^l_r,k^l_r\} \to lab_{\widehat{R}}(e_l)
	}
	$$
	Here $\DPO(r_0) = \times(\widehat{L})\div (\widehat{R} + \$_0^\bullet)$; $E_{\widehat{R}} = \{e_1,\dotsc,e_l\}$; $b^\prime_r=b_r$ for all $r \in P_N$, $k^\prime_r = k_r$ for $r \ne r_0$ and $k^\prime_{r_0} = k_{r_0}-1$; for all $r \in P_N$ it holds that $b^\prime_r = b^0_r+\dotsc+b^l_r$, $k^\prime_r = k^0_r+\dotsc+k^l_r$; and $$Y^\prime = Y^0\left[e_0/\widehat{R}\left[e_1/Y^1,\dotsc, e_l/Y^l\right]\right].$$
	
	This is just an explicit general form of the rule application of $(\div\to)$ possible for the given sequent. Remark \ref{remark_invertibility} implies that the sequent $Y^0\left[e_0/\widehat{L}\right]+\sigma\{b^0_r,k^0_r\} \to X$ is also derivable. Applying the induction hypothesis, we obtain that $X^\bullet \Rightarrow^\ast Y^0\left[e_0/\widehat{L}\right]$. Then, $Y^0\left[e_0/\widehat{L}\right] \underset{r_0}{\Rightarrow} Y^0\left[e_0/\widehat{R}\right]$. Finally, let us apply the induction hypothesis to the sequent $Y^i+\sigma\{b^i_r,k^i_r\} \to lab_{\widehat{R}}(e_i)$ for each $i=1,\dotsc,l$, therefore concluding that $lab_{\widehat{R}}(e_i)^\bullet \Rightarrow^\ast Y^i$. Now it remains to start with the hypergraph $Y^0\left[e_0/\widehat{R}\right]$ and to successively remake each derivation $lab_{\widehat{R}}(e_i)^\bullet \Rightarrow^\ast Y^i$ within it (for $i=1,\dotsc,l$), thus replacing each hyperedge $e_i$ of $\widehat{R}$ by $Y^i$. Therefore, we obtain that 
	$$
	X^\bullet \Rightarrow^\ast Y^0\left[e_0/\widehat{L}\right] \Rightarrow Y^0\left[e_0/\widehat{R}\right] \Rightarrow^\ast Y^0\left[e_0/\widehat{R}\left[e_1/Y^1,\dotsc, e_l/Y^l\right]\right] = Y^\prime.
	$$
	
	\textit{Case 2.} The last rule applied is either (i) the rule $(!\to)$, or (ii) the rule $(\mathrm{w})$, or (iii) the rule $(\mathrm{c})$. Assume that the main type of this rule application is $!\DPO(r_0)$ for $r_0 \in P_N$. Then it must be of the form
	
	$$
	\infer[]{
		Y^\prime+\sigma\{b_r,k_r\} \to X
	}{
		Y^\prime+\sigma\{b^\prime_r,k^\prime_r\} \to X
	}
	$$
	Here $b^\prime_r=b_r$, $k^\prime_r = k_r$ for $r\ne r_0$; in case (i), $b^\prime_{r_0} = b_{r_0}-1$, $k^\prime_{r_0} = k_{r_0}+1$; in case (ii), $b^\prime_{r_0} = b_{r_0}-1$, $k^\prime_{r_0} = k_{r_0}$; in case (i), $b^\prime_{r_0} = b_{r_0}+1$, $k^\prime_{r_0} = k_{r_0}$. The induction hypothesis completes the proof in all the three subcases.
\end{proof}

Now we introduce an $\HMEL$-grammar corresponding to a given DPO grammar. 
\begin{construction}\label{construction_hmel_from_dpo}
	Let $HGr=\langle N, \Sigma, P, S^\bullet \rangle$ be a normalized DPO grammar. Then $\mathrm{LG}(HGr) = \langle \Sigma, S^\prime, \triangleright \rangle$ where $\triangleright$ consists of pairs $a \triangleright T_a$, and
	$
	S^\prime = S \div \left(\sum\limits_{r \in P_N} \left(!\DPO(r)\right)^\bullet+\$_0^\bullet\right).
	$
	\\
	Here we apply the exponential to each type $\DPO(r)$ (for $r \in P_N$) and store the result in $S^\prime$.
\end{construction}
This construction turns out to be a straightforward generalization of the one from \cite{Kanazawa99}. Using it we can prove the following lemma:
\begin{lemma}\label{th_dpo_to_hmel}
	If $HGr$ is a normalized DPO grammar, then $L(\mathrm{LG}(HGr)) = L(HGr)$.
\end{lemma}
\begin{proof}
	A hypergraph $H \in \mathcal{H}(\Sigma)$ belongs to $L(\mathrm{LG}(HGr))$ if and only if $\HMEL \vdash t(H) \to S^\prime$ where $t(H)$ is the relabeling of $H$ such that $t(e) = T_{lab_H(e)}$. This sequent is equiderivable with the sequent $t(H) + \sum\limits_{r \in P_N} \left(!\DPO(r)\right)^\bullet \to S$ according to Remark \ref{remark_invertibility} for $\HMEL$. According to Lemmas \ref{lemma_dpo_many} and \ref{lemma_dpo_main}, this happens if and only if $S^\bullet \Rightarrow^\ast t(H)$ by means of nonterminal rules of $HGr$. It is straightforward to show that this is equivalent to the fact that $H \in L(HGr)$.
\end{proof}
Finally we come up with the following result:
\begin{theorem}
	$\HMEL$-grammars are at least as expressive as DPO grammars.
\end{theorem}

\section{DPO Grammars With Linear Restriction and HL-Grammars}\label{sec_lin_dpo_to_hl}

The above construction strongly relies on the exponential modality, which enables one to compress several copies of the same type in the antecedent into a single one. In the case of $\HL$ where we do not have the $!$ modality, Construction \ref{construction_hmel_from_dpo} does not work anymore. Moreover, it is clear that it must fail: languages generated by $\HL$-grammars are decidable and even are in NP while DPO grammars generate at least all recursively enumerable string languages, so the two classes of grammars do not generate the same languages. However, Construction \ref{construction_dpo_rule} still works since it uses only $\times$ and $\div$. It facilitates to link $\HL$-grammars with DPO grammars with a certain restriction formally defined below.
\begin{definition}
	Given a DPO hypergraph grammar $HGr = \langle N, \Sigma, P, S^\bullet \rangle$, let $L_c(HGr)$ consist of all hypergraphs $H \in L(HGr)$ such that there exists a derivation $S^\bullet \Rightarrow^\ast H$ with no more than $c\cdot |E_H|$ steps.
\end{definition}
\begin{construction}\label{construction_truncated_dpo_grammar}
	Let $HGr=\langle N, \Sigma, P, S^\bullet \rangle$ be a normalized grammar; let $c \in \mathbb{N}$. Then we construct an $\HL$-grammar $\mathrm{LG}_c(HGr) = \langle \Sigma, S, \triangleright \rangle$ where $\triangleright$ contains exactly the following pairs: 
	$$
	a \triangleright \times\left(T_a^\bullet + \sum\limits_{r \in P_N} k_{r}\cdot \DPO(r)^\bullet\right)
	\mbox{ for }
	k_r \ge 0, \; \sum\limits_{r \in P_N} k_{r} \le c.
	$$
	Note that $\triangleright$ is a finite relation since there are finitely many $k_r \in \mathbb{N}$ satisfying the above requirements.
\end{construction}

\begin{example}\label{ex_dpo_to_hl}
	Consider a DPO grammar $HGr = \langle N, \Sigma,P, S^\bullet\rangle$ where $N = \{S\}$, $\Sigma = \{a\}$, and $P = \{r_1,r_2,r_3\}$:
	\begin{enumerate}
		\item $r_1 = \left(\vcenter{\hbox{{\tikz[baseline=.1ex]{
						\node[hyperedge] {$S$}
		}}}} \;\leftarrow\; \Dis[0] \;\rightarrow\; \Dis[0] \right)$;
		\item $r_2 = \left( \vcenter{\hbox{{\tikz[baseline=.1ex]{
						\node[hyperedge] {$S$};
		}}}} \;\leftarrow\; \Dis[0] \;\rightarrow\;
		\vcenter{\hbox{{\tikz[baseline=.1ex]{
						\node[hyperedge] (E) {$S$};
						\node[node,right=3mm of E] (N) {};
		}}}} \right)$;
		\item $r_3 = \left( \vcenter{\hbox{{\tikz[baseline=.1ex]{
						\node[node, label=left:{\scriptsize $\{ 1 \}$}] (N1) {};
						\node[node, right = 5mm of N1, label=right:{\scriptsize $\{ 2 \}$}] (N2) {};
		}}}} \;\leftarrow\; \Dis[2] \;\rightarrow\; 
		\vcenter{\hbox{{\tikz[baseline=.1ex]{
						\node[node, label=left:{\scriptsize $\{ 1 \}$}] (N1) {};
						\node[node, right = 5mm of N1, label=right:{\scriptsize $\{ 2 \}$}] (N2) {};
						\draw[->,black] (N1) -- node[above] {$a$} (N2);
		}}}} \right)$.
	\end{enumerate}
	It is not hard to see that it generates all graphs (with edges having two attachment nodes): the rule $r_2$ produces nodes while $r_3$ produces edges. Consider e.g. the following derivation:
	\begin{equation}\label{eq_dpo_der}
		\vcenter{\hbox{{\tikz[baseline=.1ex]{
						\node[hyperedge] {$S$}
		}}}}
		\;\;\underset{r_2}{\Rightarrow}\;\;
		\vcenter{\hbox{{\tikz[baseline=.1ex]{
						\node[hyperedge] (E) {$S$};
						\node[node,right=3mm of E] (N) {};
		}}}}
		\;\;\underset{r_2}{\Rightarrow}\;\;
		\vcenter{\hbox{{\tikz[baseline=.1ex]{
						\node[hyperedge] (E) {$S$};
						\node[node,right=3mm of E] (N1) {};
						\node[node,right=3mm of N1] (N2) {};
		}}}}
		\;\;\underset{r_1}{\Rightarrow}\;\;
		\vcenter{\hbox{{\tikz[baseline=.1ex]{
						\node[node] (N1) {};
						\node[node, right = 5mm of N1] (N2) {};
		}}}}
		\;\;\underset{r_3}{\Rightarrow}\;\;
		\vcenter{\hbox{{\tikz[baseline=.1ex]{
						\node[node] (N1) {};
						\node[node, right = 5mm of N1] (N2) {};
						\draw[->,black] (N1) -- node[above] {$a$} (N2);
		}}}}
		\;\;\underset{r_3}{\Rightarrow}\;\;
		\vcenter{\hbox{{\tikz[baseline=.1ex]{
						\node[node] (N1) {};
						\node[node, right = 5mm of N1] {} edge [in=45,out=-45,loop] (N2);
						\node[right=2.5mm of N2] (T) {$a$};
						\draw[->,black] (N1) -- node[above] {$a$} (N2);
		}}}}
		\;\;\underset{r_3}{\Rightarrow}\;\;
		\vcenter{\hbox{{\tikz[baseline=.1ex]{
						\node[node] {} edge [in=135,out=-135,loop] (N1);
						\node[left=2.5mm of N1] (T1) {$a$};
						\node[node, right = 5mm of N1] {} edge [in=45,out=-45,loop] (N2);
						\node[right=2.5mm of N2] (T2) {$a$};
						\draw[->,black] (N1) -- node[above] {$a$} (N2);
		}}}}
	\end{equation}
	\noindent 
	Note that $HGr$ is not normalized; using Construction \ref{construction_normalized} we replace $r_3$ by the following two rules:
	
	%%%Normalisation of ALL
	\begin{enumerate}
		\item $r_3^\prime = \left( \vcenter{\hbox{{\tikz[baseline=.1ex]{
						\node[node, label=left:{\scriptsize $\{ 1 \}$}] (N1) {};
						\node[node, right = 5mm of N1, label=right:{\scriptsize $\{ 2 \}$}] (N2) {};
		}}}} \;\leftarrow\; \Dis[2] \;\rightarrow\; 
		\vcenter{\hbox{{\tikz[baseline=.1ex]{
						\node[node, label=left:{\scriptsize $\{ 1 \}$}] (N1) {};
						\node[node, right = 5mm of N1, label=right:{\scriptsize $\{ 2 \}$}] (N2) {};
						\draw[->,black] (N1) -- node[above] {$T_a$} (N2);
		}}}} \right)$;
		\item $r_3^{\prime\prime} = \left( \vcenter{\hbox{{\tikz[baseline=.1ex]{
						\node[node, label=left:{\scriptsize $\{ 1 \}$}] (N1) {};
						\node[node, right = 5mm of N1, label=right:{\scriptsize $\{ 2 \}$}] (N2) {};
						\draw[->,black] (N1) -- node[above] {$T_a$} (N2);
		}}}} \;\leftarrow\; \Dis[2] \;\rightarrow\; 
		\vcenter{\hbox{{\tikz[baseline=.1ex]{
						\node[node, label=left:{\scriptsize $\{ 1 \}$}] (N1) {};
						\node[node, right = 5mm of N1, label=right:{\scriptsize $\{ 2 \}$}] (N2) {};
						\draw[->,black] (N1) -- node[above] {$a$} (N2);
		}}}} \right)$.
	\end{enumerate}
	\noindent
	Let us denote this new normalized grammar $HGr^\prime$. Then we convert its nonterminal rules into types using Construction \ref{construction_dpo_rule}:	
	\begin{enumerate}
		\item $X_1 = \DPO(r_1) = \times\left( \vcenter{\hbox{{\tikz[baseline=.1ex]{
						\node[hyperedge] (E) {$S$};
		}}}} \right) \div \left( \vcenter{\hbox{{\tikz[baseline=.1ex]{
						\node[hyperedge] {$\$_0$};
		}}}} \right)$;
		\item $X_2 = \DPO(r_2) = \times\left( \vcenter{\hbox{{\tikz[baseline=.1ex]{
						\node[hyperedge] (E) {$S$};
		}}}} \right) \div \left( \vcenter{\hbox{{\tikz[baseline=.1ex]{
						\node[hyperedge] (E) {$S$};
						\node[node,right=3mm of E] (N) {};
						\node[hyperedge,right=3mm of N] {$\$_0$};
		}}}} \right)$;
		\item $X_3 = \DPO(r_3^\prime) = \times\left( \vcenter{\hbox{{\tikz[baseline=.1ex]{
						\node[node, label=left:{\scriptsize $(1)$}] (N1) {};
						\node[node, right = 5mm of N1, label=right:{\scriptsize $(2)$}] (N2) {};
		}}}}\right)\div\left( \vcenter{\hbox{{\tikz[baseline=.1ex]{
						\node[node, label=left:{\scriptsize $(1)$}] (N1) {};
						\node[node, right = 5mm of N1, label=right:{\scriptsize $(2)$}] (N2) {};
						\node[hyperedge,above right=-2mm and 8mm of N2] {$\$_0$};
						\draw[->,black] (N1) -- node[above] {$T_a$} (N2);
		}}}} \right)$;
	\end{enumerate}
	\noindent
	Finally, we introduce an $\HL$-grammar $\mathrm{LG}_2(HGr^\prime) = \langle \Sigma, S, \triangleright \rangle$ according to Construction \ref{construction_truncated_dpo_grammar}. The binary relation $\triangleright$ consists of the following 13 pairs (in fact, of the 10 distinct pairs) where $i,j \in \{1,2,3\}$:
	\begin{multicols}{2}
		\begin{itemize}
			\item $a \triangleright T = \times\left( \vcenter{\hbox{{\tikz[baseline=.1ex]{
							\node[node, label=left:{\scriptsize $(1)$}] (N1) {};
							\node[node, right = 5mm of N1, label=right:{\scriptsize $(2)$}] (N2) {};
							\draw[->,black] (N1) -- node[above] {$T_a$} (N2);
			}}}}  \right)$;
			\item $a \triangleright T_i = \times\left( \vcenter{\hbox{{\tikz[baseline=.1ex]{
							\node[node, label=left:{\scriptsize $(1)$}] (N1) {};
							\node[node, right = 5mm of N1, label=right:{\scriptsize $(2)$}] (N2) {};
							\node[hyperedge,right=8mm of N2] {$X_i$};
							\draw[->,black] (N1) -- node[above] {$T_a$} (N2);
			}}}}  \right)$;
			\item $a \triangleright T_{ij} = \times\left( \vcenter{\hbox{{\tikz[baseline=.1ex]{
							\node[node, label=left:{\scriptsize $(1)$}] (N1) {};
							\node[node, right = 5mm of N1, label=right:{\scriptsize $(2)$}] (N2) {};
							\node[hyperedge,right=8mm of N2] (E1) {$X_i$};
							\node[hyperedge,right=3mm of E1] (E2) {$X_j$};
							\draw[->,black] (N1) -- node[above] {$T_a$} (N2);
			}}}}  \right)$.
		\end{itemize}
	\end{multicols}
	Recall that in order to check that a hypergraph belongs to $L(\mathrm{LG}_2(HGr^\prime))$ we need 1) to replace labels of its hyperedges by types corresponding to them via $\triangleright$; 2) to construct a sequent with the resulting hypergraph in the antecedent and with $S$ in the succedent; 3) to derive this sequent. Let us check, for example, that $H=\vcenter{\hbox{{\tikz[baseline=.1ex]{
					\node[node] {} edge [in=135,out=-135,loop] (N1);
					\node[left=2.5mm of N1] (T1) {$a$};
					\node[node, right = 5mm of N1] {} edge [in=45,out=-45,loop] (N2);
					\node[right=2.5mm of N2] (T2) {$a$};
					\draw[->,black] (N1) -- node[above] {$a$} (N2);
	}}}}$ belongs to $L(\mathrm{LG}_2(HGr^\prime))$. We replace each label $a$ by one of the types $T$, $T_i$, or $T_{ij}$ as follows: $\vcenter{\hbox{{\tikz[baseline=.1ex]{
	\node[node] {} edge [in=135,out=-135,loop] (N1);
	\node[left=2.5mm of N1] (T1) {$T_{22}$};
	\node[node, right = 5mm of N1] {} edge [in=45,out=-45,loop] (N2);
	\node[right=2.5mm of N2] (T2) {$T_{33}$};
	\draw[->,black] (N1) -- node[above] {$T_{13}$} (N2);
	}}}}$ (compare the indices of types with the numbers of rules applied in (\ref{eq_dpo_der})). Then, it remains to check that $\vcenter{\hbox{{\tikz[baseline=.1ex]{
	\node[node] {} edge [in=135,out=-135,loop] (N1);
	\node[left=2.5mm of N1] (T1) {$T_{22}$};
	\node[node, right = 5mm of N1] {} edge [in=45,out=-45,loop] (N2);
	\node[right=2.5mm of N2] (T2) {$T_{33}$};
	\draw[->,black] (N1) -- node[above] {$T_{13}$} (N2);
	}}}} \to S$ is derivable:
	$$
	\infer[(\times\to)]{
		\vcenter{\hbox{{\tikz[baseline=.1ex]{
						\node[node] {} edge [in=135,out=-135,loop] (N1);
						\node[left=2.5mm of N1] (T1) {$T_{22}$};
						\node[node, right = 5mm of N1] {} edge [in=45,out=-45,loop] (N2);
						\node[right=2.5mm of N2] (T2) {$T_{33}$};
						\draw[->,black] (N1) -- node[above] {$T_{13}$} (N2);
		}}}} \quad\to\quad S
	}{
		\infer[(\times\to)]{
			\vcenter{\hbox{{\tikz[baseline=.1ex]{
							\node[node] {} edge [in=135,out=-135,loop] (N1);
							\node[left=2.5mm of N1] (T1) {$T_a$};
							\node[node, right = 5mm of N1] {} edge [in=45,out=-45,loop] (N2);
							\node[right=2.5mm of N2] (T2) {$T_{33}$};
							\draw[->,black] (N1) -- node[above] {$T_{13}$} (N2);
			}}}}
			\;\;\vcenter{\hbox{{\tikz[baseline=.1ex]{\node[hyperedge] {$X_2$};}}}}
			\;\;\vcenter{\hbox{{\tikz[baseline=.1ex]{\node[hyperedge] {$X_2$};}}}} 
			\quad\to\quad S
		}{
			\infer[(\times\to)]{
				\vcenter{\hbox{{\tikz[baseline=.1ex]{
								\node[node] {} edge [in=135,out=-135,loop] (N1);
								\node[left=2.5mm of N1] (T1) {$T_a$};
								\node[node, right = 5mm of N1] {} edge [in=45,out=-45,loop] (N2);
								\node[right=2.5mm of N2] (T2) {$T_{33}$};
								\draw[->,black] (N1) -- node[above] {$T_{a}$} (N2);
				}}}}
				\;\;\vcenter{\hbox{{\tikz[baseline=.1ex]{\node[hyperedge] {$X_2$};}}}}
				\;\;\vcenter{\hbox{{\tikz[baseline=.1ex]{\node[hyperedge] {$X_2$};}}}} 
				\;\;\vcenter{\hbox{{\tikz[baseline=.1ex]{\node[hyperedge] {$X_1$};}}}}
				\;\;\vcenter{\hbox{{\tikz[baseline=.1ex]{\node[hyperedge] {$X_3$};}}}} 
				\quad\to\quad S
			}{
				\vcenter{\hbox{{\tikz[baseline=.1ex]{
								\node[node] {} edge [in=135,out=-135,loop] (N1);
								\node[left=2.5mm of N1] (T1) {$T_a$};
								\node[node, right = 5mm of N1] {} edge [in=45,out=-45,loop] (N2);
								\node[right=2.5mm of N2] (T2) {$T_a$};
								\draw[->,black] (N1) -- node[above] {$T_{a}$} (N2);
				}}}}
				\;\;\vcenter{\hbox{{\tikz[baseline=.1ex]{\node[hyperedge] {$X_2$};}}}}
				\;\;\vcenter{\hbox{{\tikz[baseline=.1ex]{\node[hyperedge] {$X_2$};}}}} 
				\;\;\vcenter{\hbox{{\tikz[baseline=.1ex]{\node[hyperedge] {$X_1$};}}}}
				\;\;\vcenter{\hbox{{\tikz[baseline=.1ex]{\node[hyperedge] {$X_3$};}}}}
				\;\;\vcenter{\hbox{{\tikz[baseline=.1ex]{\node[hyperedge] {$X_3$};}}}}
				\;\;\vcenter{\hbox{{\tikz[baseline=.1ex]{\node[hyperedge] {$X_3$};}}}} 
				\quad\to\quad S
			}
		}
	}
	$$
	The uppermost sequent in the above derivation is derivable, which follows from the proof of Lemma \ref{lemma_dpo_many}. This justifies that $H \in L(\mathrm{LG}_2(HGr^\prime))$.
\end{example}
The above example illustrates the following theorem, the main one in this work:
\begin{theorem}
	For each DPO grammar $HGr$ and each $c$ there is an $\HL$-grammar generating $L_c(HGr)$.
\end{theorem}
\begin{lemma}\label{th_truncated_dpo_to_hl}
	If $HGr$ is a normalized DPO grammar and $1 \le c \in \mathbb{N}$, then $L(\mathrm{LG}_{c-1}(HGr)) = L_{c}(HGr)$.
\end{lemma}
The inclusion $L(\mathrm{LG}_{c-1}(HGr)) \supseteq L_{c}(HGr)$ is proved by using Lemma \ref{lemma_dpo_many} in the same way as in Example \ref{ex_dpo_to_hl}. The other inclusion is proved by using Lemma \ref{lemma_dpo_main} and by noticing that the number of $\DPO(r)$-labeled hyperedges in an antecedent here strictly corresponds to the number of rule applications in $HGr$ (since we do not have types with the exponential).

The idea of Construction \ref{construction_truncated_dpo_grammar} and of Theorem \ref{th_truncated_dpo_to_hl} is that we store $\DPO(r)$-labeled hyperedges in each type corresponding to a terminal symbol (since we cannot store them in the succedent as in Construction \ref{construction_hmel_from_dpo}). Then, for $G \in L(\mathrm{LG}_{c}(HGr))$, after we replace each symbol $a$ in $G$ by a type $\times\left(T_a^\bullet + \sum\limits_{r \in P_N} k_{r}\cdot \DPO(r)^\bullet\right)$ for some $k_r$ and start deriving a corresponding sequent, these hyperedges eventually appear in the antecedent where they play their role shown in Lemma \ref{lemma_dpo_many}. The total number of these hyperedges, however, is limited by the number of hyperedges in $G$, hence the language $L_c(HGr)$ is generated instead of $L(HGr)$. 

Theorem \ref{th_truncated_dpo_to_hl} says that $\HL$-grammars are powerful enough to generate hypergraphs of a language generated by a DPO grammar such that the number of steps in their derivation is bounded by a linear function of the number of hyperedges. It might be the case for a DPO grammar $HGr$ that $L(HGr)=L_c(HGr)$ for some $c\in \mathbb{N}$; in fact, we claim that for each $\HL$-grammar $HLGr=\langle \Sigma,S,\triangleright \rangle$ there is a DPO grammar $HGr$ and $c\in\mathbb{N}$ such that $L(HLGr) = L_c(HGr)=L(HGr)$, although we do not prove this here (this should be a matter of another paper). In general, however, $L(HGr)\ne L_c(HGr)$ (e.g. in Example \ref{ex_dpo_to_hl} $L_{k+1}(HGr)$ contains only graphs $G$ such that $|V_G|< k\cdot |E_G|$).

\section{HL with the Conjunctive Kleene Star}\label{sec_cks}

In Section \ref{sec_hmel}, we added the exponential modality to $\HL$, which helped us to prove Theorem \ref{th_dpo_to_hmel}. Another operation, which behaves similarly, is Kleene star. The Lambek calculus with Kleene star is studied by several researchers, in particular, in \cite{Kuznetsov21}. An extension of $\mathrm{L}$ by intersection, union, and Kleene star is known as infinitary action logic \cite{Palka07}, or, in the algebraic setting, as the logic of action algebras. Note that Kleene star can be described in terms of actions within a transition system: if $A$ is a class of actions, then $A^\ast$ means actions from $A$ repeated several times \cite{Kuznetsov21}. This understanding is very close to what we use in Construction \ref{construction_hmel_from_dpo} since our goal is to be able to apply encoded DPO rules arbitrarily many times within $\HL$. Hence, let us now look at the Lambek calculus with the unit and Kleene star $\LIKS$ and, more specifically, at the rules for Kleene star \cite{Kuznetsov21}:
$$
\infer[(\to^\ast),\; n \ge 0]{\Pi \to A^\ast}{\Pi \to A^n}
\qquad\qquad
\infer[(^\ast\to)_\omega]{\Gamma, A^\ast, \Delta \to B}{\left( \Gamma, A^n, \Delta \to B \right)_{n=0}^\infty}
$$
Here $A^0 \eqdef \mathbf{1}$, $A^{n+1}\eqdef A^n\cdot A$; $\mathbf{1}$ is the unit of the product satisfying the axiom $A\cdot \mathbf{1} \leftrightarrow A$. Note that the rule $(\to^\ast)$ is in fact a countable set of rules for each $n \ge 0$; contrarily, $(^\ast\to)_\omega$ is a sinlge rule with countably many premises. Let us clarify the notion of being derivable in this calculus: the set of derivable sequents in $\LIKS$ is the least set $S$ containing all axioms of $\LIKS$ (i.e., all sequents of the form $A \to A$ and the sequent $\to \I$) such that it is closed under applications of all inference rules (i.e., if, for some rule, all sequents above the line belong to $S$, then the sequent below the line must also belong to $S$). In other words, a derivation in $\LIKS$ is again a sequence of rule applications, which now can be countable in size but which does not have branches of infinite length.

Unfortunately, the rules for Kleene star work in an undesirable way: they allow unlimited copying types in succedents of sequents (namely, if we have $n$ copies of $A$ in a succedent, then we can wrap them into a single type $A^\ast$) but not in antecedents. This motivates us to consider an operation behaving dually:
$$
\infer[(\to^\ast)_\omega]{\Pi \to \prescript{\ast}{}{A}}{\left( \Pi \to A^n \right)_{n=0}^\infty }
\qquad\qquad
\infer[(^\ast\to),\; n \ge 0]{\Gamma, \prescript{\ast}{}{A}, \Delta \to B}{\Gamma, A^n, \Delta \to B }
$$
We call the operation $\prescript{\ast}{}{A}$ \emph{the conjunctive Kleene star}. Algebraically, it can be defined in complete residuated lattices using infinitary conjunction as $\prescript{\ast}{}{a} = \bigwedge\limits_{n=0}^\infty a^n = \inf\{a^n \mid n \in \mathbb{N}\}$ (this is why we call it conjunctive). Note that the language semantics of this operation is poor: if $L$ is a language ($L \subseteq \Sigma^\ast$), $\mathbf{1}$ equals $\{\Lambda\}$, and multiplication of languages means pairwise concatenation of their words while conjunction means intersection, then $\prescript{\ast}{}{L}=\{\Lambda\}$ if $\Lambda \in L$ and $\prescript{\ast}{}{L}= \emptyset$ otherwise.

Unfortunately, we have found little about this operation in the literature. Nevertheless, it is mentioned in \cite{Montagna04} in the context of storage operators. A storage operator $I$ in an MTL-algebra works as follows: $I(a)$ is the greatest idempotent among those not greater than $a$ (it must exist in an MTL-algebra with storage). It can be shown that, if $\inf\{a^n \mid n \in \mathbb{N}\}$ exists, then it equals $I(a)$, so these operators are very close. In \cite{Montagna04}, it is mentioned that the storage operator has many analogies with Girard's exponential $!$.

Let us show how to generalize the conjunctive Kleene star and inference rules for it to hypergraphs. A question arises: how should one understand an iteration of a type, namely, $A^n$? In the string case, this means repeating a type $n$ times and writing copies in line connecting them by $\cdot$. We need to extend  this iteration procedure to hypergraphs. We suggest the following general definitions:
\begin{definition}
	A \emph{template} $T$ of rank $k$ is a hypergraph $T = \langle V_T,[2],att_T,lab_T,ext_T\rangle$ such that $rk_T(1)=rk_T(2)=rk(T)=k$. In other words, $T$ has two hyperedges of the same rank, which coincides with the rank of $T$. Hereinafter $T(H_1,H_2)$ is a shorthand notation for $T[1/H_1,2/H_2]$.
\end{definition}
\begin{definition}
	A template $T$ of rank $k$ is \emph{monoidal} if for all hypergraphs $A,B,C$ of rank $k$ it holds that 1. $T(A,T(B,C)) \cong T(T(A,B),C)$, 2. a hypergraph $U_T$ of rank $k$ exists such that $T(U_T,A)\cong T(A,U_T) \cong A$.
\end{definition}
\begin{definition}
	The \emph{$T$-iteration} $T^n(A)$ of a type $A$ (where $T$ is a monoidal template) such that $rk(A)=rk(T)$ is defined as follows: $T^0(A) \eqdef U_T$; $T^{n+1}(A) \eqdef T(T^n(A),A^\bullet)$ (for $n \ge 0$).
\end{definition}
\begin{example}
	Two examples of monoidal templates are
	\begin{itemize}
		\item $O = \vcenter{\hbox{{\tikz[baseline=.1ex]{
						\node[hyperedge] (E1) {$X$};
						\node[hyperedge,right=3mm of E1] (E2) {$X$};
		}}}}$ (i.e., $V_{O} = \emptyset$, $E_{O} = [2]$, $att_{O}(1) = att_{O}(2) = ext_{O}=\Lambda$). Note that $U_O = \Dis[0]$.
		\item $Str = \vcenter{\hbox{{\tikz[baseline=.1ex]{
						\node[node, label=left:{\scriptsize $(1)$}] (N1) {};
						\node[node,right=7mm of N1] (N2) {};
						\node[node,right=7mm of N2, label=right:{\scriptsize $(2)$}] (N3) {};
						\draw[>=stealth,->,black] (N1) -- node[above] {$Y$} (N2);
						\draw[>=stealth,->,black] (N2) -- node[above] {$Y$} (N3);
		}}}}$.
	\end{itemize}
	Here $X,Y$ are arbitrary labels, they do not matter. Note that 	$O(H,G) = H+G$ for zero-rank $H,G$. Consequently, $O^m(A) = m\cdot A^\bullet$ (where $rk(A)=0$).
\end{example}

Using monoidal templates we can define the hypergraph conjunctive Kleene star. Types of the hypergraph Lambek calculus with the conjunctive Kleene star $\prescript{\ast}{}{\HL}_{\omega}$ are built as described in Section \ref{ssec_HL_HLG} but we add one more item to the definition: if $A$ is a type such that $rk(A)=n$ and if $T$ is a monoidal template of rank $n$, then $\prescript{\ast}{T}{A}$ is also a type of rank $n$. We also add two inference rules for the new operation:
$$
\infer[(\to^\ast)_\omega]{H \to \prescript{\ast}{T}{A}}{\left( H \to \times(T^n(A)) \right)_{n=0}^\infty }
\qquad\qquad
\infer[(^\ast\to),\; n \ge 0]{G[e/\left(\prescript{\ast}{T}{A}\right)^\bullet] \to B}{G[e/T^n(A)] \to B}
$$
Usual logical questions concerning $\prescript{\ast}{}{\HL}_{\omega}$ arise. In particular, the cut elimination theorem can be proved:
\begin{theorem}\label{theorem_cut}
	If $\prescript{\ast}{}{\HL}_{\omega} \vdash H\to A$ and $\prescript{\ast}{}{\HL}_{\omega} \vdash G[e/A^\bullet]\to B$, then $\prescript{\ast}{}{\HL}_{\omega} \vdash G[e/H]\to B$.
\end{theorem}
The theorem is proved by a transfinite induction in a similar way to that from \cite{Palka07}.

Note that we can define \emph{the hypergraph Kleene star} generalizing $A^\ast$ studied in \cite{Kuznetsov21,Palka07} in the same way as $\prescript{\ast}{T}{A}$. Even if the conjunctive Kleene star is something weird and useless, we think that the definitions of template, of $T$-iteration and so on are useful in the respect that using them we could define and study the hypergraph Kleene star.

Returning to $\prescript{\ast}{}{\HL}_\omega$, we can now define $\prescript{\ast}{}{\HL}_\omega$-grammars and repeat Construction \ref{construction_hmel_from_dpo}:
\begin{construction}
	Let $HGr=\langle N, \Sigma, P, S^\bullet \rangle$ be a normalized grammar. Then $\mathrm{LG}_\omega(HGr) = \langle \Sigma, S^\prime, \triangleright \rangle$ where $\triangleright$ consists of pairs $a \triangleright T_a$, and
	$
	S^\prime = S \div \left(\sum\limits_{r \in P_N} \left(\prescript{\ast}{O}{\DPO(r)}\right)^\bullet+\$_0^\bullet\right).
	$
	\\
	Here we apply the hypergraph conjunctive Kleene star to each type $\DPO(r)$ (for $r \in P_N$) and store the result in $S^\prime$. This trick enables us to prove the following
\end{construction}
\begin{theorem}\label{th_dpo_to_hl_dks}
	If $HGr$ is a normalized DPO grammar, then $L(\mathrm{LG}_\omega(HGr)) = L(HGr)$.
\end{theorem}
%This, of course, implies undecidability of $\prescript{\ast}{}{\HL}_\omega$.

\section{Conclusion}\label{sec_conclusion}
In this work, we have made several steps in relating DPO grammars, which represent the rule-based approach in the field of graph grammars, and HL-grammars, which are an ambassador of the type-logical approach. We have proved that any DPO grammar can be transformed into an equivalent $\HMEL$-grammar using a generalization of the method from \cite{Kanazawa92}. Then we restricted derivations in DPO grammars and proved that grammars with this restriction can be converted into equivalent $\HL$-grammars. This restriction is a promising tool in our opinion: we claim that linearly-restricted DPO grammars are equivalent to $\HL$-grammars. Recently we have successfully applied the same idea of
imposing a linear restriction to solve an open problem concerning the Lambek calculus with permutation \cite{Pshenitsyn22_preprint}. In the future, we are going to establish a precise connection between $\HL$-grammars and DPO grammars.

\section*{Acknowledgments}
I thank prof. Mati Pentus and Stepan L. Kuznetsov for fruitful discussions, and anonymous reviewers for valuable remarks at all the rounds of reviewing (in particular, for suggesting studying connections with process algebras, which resulted in starting a research of another subject).

\bibliographystyle{eptcs}
\bibliography{FG_and_HLG_CKS}

\end{document}